\documentclass[a4paper,11pt,fleqn]{article}
\usepackage[utf8]{inputenc}

\usepackage{caption} 

\usepackage{geometry}

\usepackage{graphicx} 
\newcommand{\itangle}[1]{\rotatebox{-10}{\scalebox{1}[0.9]{$#1$}}}

\usepackage{amsmath}
\usepackage{amsfonts,amssymb,amsthm}
\usepackage{xcolor}
\definecolor{linkColor}{RGB}{0, 128, 128}
\definecolor{citeColor}{RGB}{0, 112, 64}
\definecolor{urlColor}{RGB}{120, 0, 120}
\usepackage{hyperref}
\hypersetup{
    colorlinks=true,
    linkcolor=linkColor,
    citecolor=citeColor,      
    urlcolor=urlColor,
}
\usepackage{url}

\definecolor{oracleColor}{HTML}{A0BDD3} 
\definecolor{noiseColor}{HTML}{F2B547} 
\definecolor{cmdColor}{RGB}{0, 0, 0}

\theoremstyle{plain}
\newtheorem{thm}{Theorem}
\newtheorem{lem}[thm]{Lemma}

\newtheorem{clm}[thm]{Claim}

\theoremstyle{definition}
\newtheorem{defn}{Definition}

\usepackage{tikz}

\newcommand{\cH}{\mathcal{H}}

\renewcommand{\>}{\rangle}
\newcommand{\<}{\langle}

\newcommand{\qqAnd}{\qquad\text{and}\qquad}

\usepackage{mathrsfs} 
\newcommand{\ccO}{\mathscr{O}}

\newcommand{\ccN}{\mathscr{N}}

\newcommand{\ccI}{\mathscr{I}} 

\newcommand{\regQ}{\mathsf{Q}}
\newcommand{\regQi}{\mathsf{Q_i}}
\newcommand{\regQo}{\mathsf{Q_o}}
\newcommand{\regR}{\mathsf{R}}
\newcommand{\regT}{\mathsf{T}}
\newcommand{\regW}{\mathsf{W}}

\newcommand{\regFunc}{\regT} 

\newcommand\unif{u}

\newcommand\bin{{\mathrm{b}}}
\newcommand\Pauli{{\mathbb{P}}}

\newcommand\actv{{act}} 
\newcommand\pasv{{pas}} 

\newcommand{\OO}{\mathrm{O}}

\newcommand\progA{{\mathfrak{A}}} 
\newcommand\progB{{\mathfrak{B}}} 
\newcommand\progC{{\mathfrak{C}}} 
\newcommand{\progLab}{{\mathfrak{Lab}}}

\title{Addendum to ``Quantum Search with Noisy Oracle''}
\author{
Ansis Rosmanis\thanks{E-mail: \texttt{rosmanis@protonmail.com}}\\ [.5ex]
\normalsize  (unaffiliated)
}
\date{May 20, 2024}

\begin{document}

\maketitle

\begin{abstract}
In this note, I generalize the techniques of my recent work (arXiv:2309.14944) and show that, even if just a single known qubit of query registers is affected by the depolarizing noise of rate $p$, quantum search among $n$ elements cannot be done any faster than in $\OO(np)$ queries.
This holds both when the affected qubit is one of the $\log n$ index qubits and when it is the target qubit.
\end{abstract}

\section{Overview of the Main Result}

\subsection{Noise model}

In the original article, \cite{Rosmanis:2023:NoisyOracle}, we considered a noise that simultaneously depolarizes all query registers. The action of such a depolarizing noise commutes with the faultless oracle call. On the other hand, as we will implicitly show, the action of the depolarizing noise on a single qubit does not commute with the faultless oracle call.

 Indeed, if the single-qubit depolarizing noise is promised to happen only before or only after the faultless oracle call, quadratic quantum speedups survive. We sketch adaptations of Grover's algorithm to these scenarios in Appendix~\ref{sec:OneTwoSided}. On the other hand, the main result of this note is that, if the noise happens both before and after the faultless oracle call, those speedups vanish.
 
Therefore we assume that, independently before and after the faultless oracle call, a given qubit is completely depolarized with probability $r$ (see Figure~\ref{fig:noisy_qubit_orac} for an illustration). More formally, we define the faulty oracle call as follows.

\def\ygap{0.35}

\begin{figure}[!h]
\centering
\begin{tikzpicture}
\draw [white] (-0.1,-0.35-\ygap*2.5) rectangle (0.1,0.25+\ygap*4); 
  \foreach \y in {1,...,5}
  {
      \node [gray] at (0,\y*\ygap-\ygap) {\tiny $\y$};
   }
   \node [gray] at (0,-2.5*\ygap) {\tiny $0$};
\end{tikzpicture}
\hspace{13pt}
\begin{tikzpicture}
\draw [white] (-2.4,-0.35-\ygap*2.5) rectangle (2.6,0.25+\ygap*4); 
  \foreach \y in {1,...,5}
  {
      \draw (-2.4,\ygap*\y-\ygap)--(2.4,\ygap*\y-\ygap) ;
   }
   \draw (-2.4,-\ygap*2.5)--(2.4,-\ygap*2.5);
   \draw [draw=black,fill=oracleColor] (-0.5,-0.15-\ygap*2.5) rectangle (0.5,0.15+\ygap*4);
   \draw [fill=noiseColor, rounded corners = 1.3mm] (-1.8,-0.26-2.5*\ygap) rectangle (-0.8,0.26-2.5*\ygap);
   \draw [fill=noiseColor, rounded corners = 1.3mm] (0.8,-0.26-2.5*\ygap) rectangle (1.8,0.26-2.5*\ygap);

   \node at (0,\ygap) {$O_f$};
    \node at (-1.28,-.03-2.5*\ygap) {$\ccN_{0,r}$};
    \node at (1.32,-.03-2.5*\ygap) {$\ccN_{0,r}$};
\end{tikzpicture}
\hspace{17pt}
\begin{tikzpicture}
\draw [white] (-2.4,-0.35-\ygap*2.5) rectangle (2.6,0.25+\ygap*4); 
  \foreach \y in {1,...,5}
  {
      \draw (-2.4,\ygap*\y-\ygap)--(2.4,\ygap*\y-\ygap) ;
   }
   \draw (-2.4,-\ygap*2.5)--(2.4,-\ygap*2.5);
   \draw [draw=black,fill=oracleColor] (-0.5,-0.15-\ygap*2.5) rectangle (0.5,0.15+\ygap*4);
   \draw [fill=noiseColor, rounded corners = 1.3mm] (-1.8,-0.26+3*\ygap) rectangle (-0.8,0.26+3*\ygap);
   \draw [fill=noiseColor, rounded corners = 1.3mm] (0.8,-0.26+3*\ygap) rectangle (1.8,0.26+3*\ygap);

   \node at (0,\ygap) {$O_f$};
    \node at (-1.28,-.03+3*\ygap) {$\ccN_{4,r}$};
    \node at (1.32,-.03+3*\ygap) {$\ccN_{4,r}$};
\end{tikzpicture}

{\hspace{29pt}\small(\emph{target-qubit noise}) \hspace{78pt}
\small(\emph{index-qubit noise})}
\captionsetup{font=small}
\captionsetup{width=0.9\textwidth}
\caption[my caption]{%
The circuit diagram of the noisy oracle call $\ccO_{f,j,r}$ when the target qubit is noisy ($j=0$; on the left) and when an index qubit is noisy ($j=4$; on the right). The noise acts as the identity on all qubits except the $j$-th query qubit, which independently before and after the faultless oracle call $O_f$ is completely depolarized with probability $r$.
 The qubits of the query register $\regQ=\regQi\regQo$ are enumerated using values $j\in\{0,1,\ldots,\log n\}$ (here $n=2^5$), with $0$-th qubit forming $\regQo$ and $j$-th qubits with $j\in\{1,\ldots,\log n\}$ forming $\regQi$.
}
\label{fig:noisy_qubit_orac}
\end{figure}
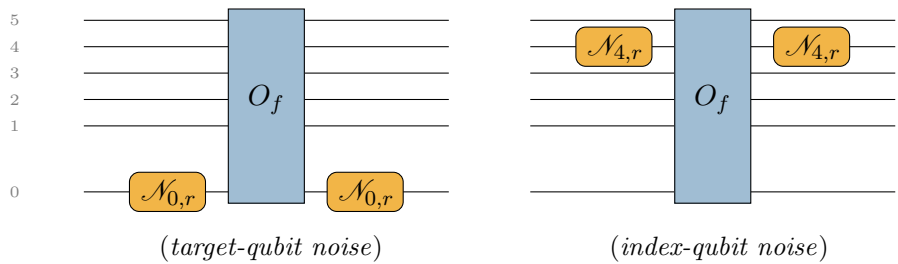

First of all, for the sake of notational convenience, in this note we consider the phase oracle rather than the standard oracle. That is, the \emph{faultless} oracle call to $f$ is
\begin{equation}
\label{eq:faultless}
O_f := \sum_{x\in[n]}\sum_{y\in\{0,1\}}(-1)^{f(x)y}|x,y\>\<x,y|;
\end{equation}
here and throughout the note we consider $f$ to be Boolean-valued. 
As before, we denote the CPTP map corresponding to $O_f$ by $\ccO_f$.

In addition, let us assume that $n$ is a power of $2$, so $O_f$ acts on $\log n+1$ qubits: the $\log n$ index qubits forming the query-input register $\regQi$ and the sole target qubit forming the query output register $\regQo$.

As in \cite[Figure 5]{Rosmanis:2023:NoisyOracle}, let $x_{\mathrm b}=x_{\mathrm b,\log n}\ldots x_{\mathrm b,2}x_{\mathrm b,1}$ denote the $\log n$-bit binary encoding of an input $x$. In particular, the bits $x_{\mathrm b,j}\in\{0,1\}$, where $j\in\{1,\ldots,\log n\}$, are such that 
$x=\sum_{j=1}^{\log n}2^{j-1}x_{\mathrm b,j}$. For $j\in\{1,\ldots,\log n\}$, let us refer to the qubit of the register $\regQi$ used for storing $x_{\mathrm b, j}$ as \emph{$j$-th query qubit}. Let us also occasionally refer to the target qubit $\regQo$ as \emph{$0$-th query qubit} in order to address the noise on any single qubit of $\regQ=\regQi\regQo$ in a consistent way.

For $j\in\{0,1,\ldots,\log n\}$, we define the depolarizing noise of $j$-th query qubit of noise rate $r$ as a CPTP map that with probability $1-r$ acts as the identity and with probability $r$ completely depolarizes $j$-th query qubit. We denote this CPTP map by $\ccN_{j,r}$.
We define the noisy oracle call as the CPTP map $\ccO_{f,j,r}:=\ccN_{j,r}\circ\ccO_f\circ\ccN_{j,r}$.%
\footnote{We could have used a similar definition in the original article, \cite{Rosmanis:2023:NoisyOracle}, where the noise $\ccN_{p}$ acting on all query registers commuted with $\ccO_f$. In particular, in that case, equating $r=1-\sqrt{1-p}\approx p/2$ would give us 
 $\ccN_{r}\circ \ccO_f \circ \ccN_{r} = \ccN_{r}\circ \ccN_{r}\circ \ccO_f = \ccN_{2r-r^2}\circ \ccO_f  = \ccO_{f,p}$.}

\subsection{Hardness of search with a noisy oracle qubit}

In this note, we show the following.

\begin{thm}
\label{thm:mainSingle}
Every algorithm that is given a known $j\in\{0,1,\ldots,\log n\}$ and an access to the noisy oracle with noisy $j$-th query qubit of noise rate $r$ (namely, $\ccO_{f,j,r}$) and that solves the unstructured search problem in the worst case setting with probability at least $1-\epsilon$ must make at least $nr(1-\epsilon)/2000-1$ queries to the noisy oracle.
\end{thm}

Let us make several remarks regarding this result.
\begin{itemize}
\item The lower bound given by Theorem~\ref{thm:mainSingle} is almost tight in $n$ and $r$. In particular, assuming $r\le 0.99$, the algorithm outlined in the original article (see \cite[Section 7]{Rosmanis:2023:NoisyOracle}) can solve the unstructured search using $\OO(\max\{nr(1+r\log n),\sqrt{n}\})$ queries to $\ccO_{f,j,r}$. 
\item Unlike in the original article \cite{Rosmanis:2023:NoisyOracle}, where the complexity of the problem remains the same even when one is given flag bits that indicate occurrences of errors, now, in the case of a single-qubit noise, such flag bits make the problem easier. In particular, if, for every noisy oracle call, we had two bits indicating the occurrence of complete depolarizing of $j$-th query qubit (i.e., the map $\ccN_{j,1}$) before and after the noiseless oracle call (see Figure~\ref{fig:noisy_qubit_orac_signal}), we could perform the unstructured search in $\OO(\max\{nr^2,\sqrt{n}\})$ queries. 
We elaborate the details in Appendix~\ref{sec:search_with_flags}.
The matching $\Omega(nr^2)$ query lower bound can be shown using techniques similar to those in \cite{Rosmanis:2023:NoisyOracle}.
\item The lower bound given by Theorem~\ref{thm:mainSingle} also applies to the stronger noise model where every qubit independently suffers the depolarizing noise of rate $r$, that is,  the noisy oracle call is 
$\bigotimes_j\ccN_{j,r}\circ\ccO_f\circ\bigotimes_j\ccN_{j,r}$.
Here we may consider the tensor product $\bigotimes_j\ccN_{j,r}$ to be over all qubits, not only the query qubits.
Hence, for noisy unstructured search, Theorem~\ref{thm:mainSingle} is a strict improvement upon the bound given by Chen, Cotler, Huang, and Li ~\cite{ChenCHL:2022:nisq} in the noisy intermediate-scale quantum (NISQ) computational model in the sense that Theorem~\ref{thm:mainSingle} considers a weaker noise model and that it gets rid of logarithmic factors.
\end{itemize}

\subsection{Techniques}

Both in the original article~\cite{Rosmanis:2023:NoisyOracle} and in this note, we construct purifications of noisy oracle calls. Formally speaking, these purifications are Stinespring representations of the corresponding CPTP maps.

In the original article, to purify the noisy oracle call $\ccO_{f,p}$, we first used the record register to purify the CPTP noise $\ccN_p$, after which the faultless oracle call $O_f$ was applied. Because $O_f$ is unitary, it did not affect the record register at all.

Now, however, we consider the joint ``noise+oracle'' operation as a single CPTP map, and, for it, we find a specific Kraus representation so that certain Kraus operators essentially indicate that the algorithm has found the marked element. The record register now purifies this joint CPTP noise+oracle map by ``recording" which Kraus operator has been applied. This new approach also lets one to reprove the 2008 result by Regev and Schiff \cite{regev:faultySearch}, which we do in Appendix~\ref{app:negligent}.

\subsection{Intuition behind the bound}
\label{sec:intuition}

For a single qubit, we can think of the completely depolarizing channel as an operation that first completely dephases the qubit, collapsing it to either $|0\>$ or $|1\>$, and then flips the value of the resulting bit with probability $1/2$.

\bigskip

Consider the depolarizing noise affecting the target qubit. 
In the faultless case, $O_f$ reflects about marked elements if the target qubit is set to $|1\>$ and acts as the identity if the target qubit is set to $|0\>$. However, if the target qubit gets depoarized before applying $O_f$, which happens with probability $r$, whether or not the desired reflection gets applied is out of control of the algorithm. What the algorithm could do in this case is to measure the target qubit after the noisy oracle call and see if the desired reflection was applied.\footnote{We elaborate on this in Appendix~\ref{sec:OneTwoSided}, where we consider one-sided noise.} However, still, if the target qubit gets depoarized after applying $O_f$, which, independently, also happens with probability $r$, then this error obfuscates whether the desired reflection was performed.

Hence, the algorithm cannot fully trust that the correct reflection will be applied and that it will be possible to tell if it got applied. While the probability of the errors occurring both before and after $O_f$ is only $r^2$, some error happens with probability $2r-r^2$, and the algorithm cannot tell whether it happened before or after $O_f$. 

As a concrete example, suppose the target qubit is initialized to $|1\>$, then the noisy oracle is called, and afterwards this qubit is measured to yield $|0\>$ (the measurement will yield $|0\>$ with probability $r-r^2/2$). We know that at least one error has occurred and the likelihood of having successfully performed a reflection about marked elements is exactly $50\%$.

This situation, where with probability $\Theta(r)$, instead of performing the reflection about marked elements, the oracle does not do anything and we have no ability to detect this inaction is very reminiscent of the negligent oracle studied by Regev and Schiff. Therefore, it is not surprising that we obtain the same lower bound as they do.

\bigskip

When it is an index qubit that is affected by the noise, a similar intuition holds. For simplicity, suppose there is exactly one marked element. In this scenario, the value of the qubit affected by the noise determines whether or not the state of other query qubits undergo any non-trivial reflection. The inability to know perfectly what was the value of this qubit when the faultless oracle was applied (as a part of a noisy oracle call) leads to essentially the same issues as in the scenario of the noisy target qubit.

\section{Kraus Representations of the Noisy Oracle}

As in the original article, \cite{Rosmanis:2023:NoisyOracle}, we only consider functions with exactly one marked element. Recall that we denote the function for which $x$ is the unique marked element by $f_x$. In this section, we consider Kraus representations of the noisy oracle call $\ccO_{f_x,j,p}$. In particular, we will provide a specific Kraus representation of $\ccO_{f_x,j,p}$ that we will utilize in Section~\ref{sec:progMes} to define the progress measure $\Psi_t$.

\subsection{A Kraus representation of the noise}

First, let us consider a Kraus representation of the depolarizing noise  $\ccN_{j,r}$ acting on $j$-th query qubit.
For that, recall Pauli operators
\[
\sigma_I:=\left[\begin{array}{cc}1&0\\0&1\end{array}\right], \quad
\sigma_X:=\left[\begin{array}{cc}0&1\\1&0\end{array}\right], \quad
\sigma_Y:=\left[\begin{array}{cc}0&-i\\i&0\end{array}\right], \quad
\sigma_Z:=\left[\begin{array}{cc}1&0\\0&-1\end{array}\right],
\]
and we will often refer to them jointly as $\sigma_P$, where $P\in \{I,X,Y,Z\}$. 
For conciseness, let $\Pauli:=\{I,X,Y,Z\}$, which is the set of labels of Pauli operators.
Note that 
\[
\sigma_X\sigma_Y=i\sigma_Z=-\sigma_Y\sigma_X,\qquad
\sigma_Y\sigma_Z=i\sigma_X=-\sigma_Z\sigma_Y,\qquad
\sigma_Z\sigma_X=i\sigma_Y=-\sigma_X\sigma_Z.
\]
When we want to emphasize that a Pauli operator $\sigma_P$ acts on $j$-th query qubit, we may write it as $\sigma_P^{(j)}$.

The completely depolarizing noise $\ccN_{j,1}$ has Kraus operators $\sigma_I^{(j)}\!/2$, $\sigma_X^{(j)}\!/2$, $\sigma_Y^{(j)}\!/2$, $\sigma_Z^{(j)}\!/2$. Namely, $\ccN_{j,1}$ acts on a density operator $\rho$ as 
\[
\ccN_{j,1} \colon \rho \mapsto \sum_{P\in\Pauli}  
\big(\sigma_I^{(j)}\!/2\big)\rho\big(\sigma_I^{(j)}\!/2\big)^*.
\]
Given this, we can easily obtain a Kraus representation of the depolarizing noise $\ccN_{j,r}$ of rate $r$.
Let $d_{I,r} := \sqrt{4-3r}/2$ and $d_{X,r} = d_{Y,r} = d_{Z,r} := \sqrt{r}/2$.
The depolarizing noise $\ccN_{j,r}$ has Kraus operators $d_{I,r}\sigma_I^{(j)}$, $d_{X,r}\sigma_X^{(j)}$, $d_{Y,r}\sigma_Y^{(j)}$, $d_{Z,r}\sigma_Z^{(j)}$, namely, $\ccN_{j,r}$ acts as 
\[
\ccN_{j,r} \colon \rho \mapsto 
\sum_{P\in\Pauli}
\big(d_{P,r} \sigma_P^{(j)}\big)\rho\big(d_{P,r} \sigma_P^{(j)}\big)^*.
\]

\subsection{A simple Kraus representation of the noisy oracle}
\label{sec:KrausOfNoisy}

Because the faultless oracle call $\ccO_{f_x}$ has a simple Kraus representation $\ccO_{f_x}\colon \rho\mapsto O_{f_x}\rho O_{f_x}$, by using the above Kraus representation of $\ccN_{j,r}$, we can easily obtain a Kraus representation of $\ccO_{f_x,j,r}=\ccN_{j,r}\circ\ccO_{f_x}\circ\ccN_{j,r}$.
In particular, $\ccO_{f_x,j,r}$ can be represented using sixteen Kraus operators $G_{P,P'}^{x,j} := d_{P,r} d_{P',r} \sigma_P^{(j)}O_{f_x}\sigma_{P'}^{(j)}$, where $P,P'\in\Pauli$. That is,
\[
\ccO_{f_x,j,r} \colon \rho \mapsto 
\sum_{P,P'\in\Pauli}
G_{P,P'}^{x,j}\rho (G_{P,P'}^{x,j})^*.
\]

Let us inspect Kraus operators $G_{P,P'}^{x,j}$ in more detail. A lot of analysis will be the same regardless whether the noise affects the target qubit  (i.e., $j=0$)  or it affects an index qubit (i.e., $j\in\{1,\ldots,\log n\}$). However, occasionally we will differentiate the two scenarios. In that case, for a given operator $A$, we use a caron (i.e., $\check{A}$) to indicate that we are defining this operator for the $j=0$ scenario and a circumflex (i.e., $\hat{A}$) to indicate that we are defining this operator for the $j\in\{1,\ldots,\log n\}$ scenario.

To start with, let us define the projector $\Pi_{x,j}$ and a linear isometry $\Xi_{x,j}$, both acting on register $\regQ$. Intuitively, $\Pi_{x,j}$ projects on the $(2n-2)$-dimensional subspace (of the space corresponding to $\regQ$) unaffected by the noisy oracle call, while $\Xi_{x,j}$ acts on the $2$-dimensional subspace affected by the noisy oracle call, the effect of the oracle call on the latter subspace being exactly $\Xi_{x,j}$ when $r=0$. The definitions of $\Pi_{x,j}$ and $\Xi_{x,j}$ slightly differ in the two scenarios.

\paragraph{Noisy target qubit, $j=0$.}

For every $x\in [n]$, we define 
 \[
\check\Pi_{x,j}:=I_\regQ-|x\>\<x|\otimes \sigma_I,
\qquad
\check\Xi_{x,j}:=|x\>\<x|\otimes \sigma_Z.
\]

\paragraph{Noisy index qubit, $j\in\{1,\ldots,\log n\}$.}

Given $x\in[n]$, let $x^j\in[n]$ be such that the binary representation of $x^j$ is obtained by flipping the $j$-th bit in the binary representation of $x$. In other words, we have $|x^j\>=\sigma_X^{(j)}|x\>$. For every $x\in [n]$, we define 
\[
\hat\Pi_{x,j} := I_{\regQ} - |x,1\>\<x,1| - |x^j,1\>\<x^j,1|,
\qquad
\hat\Xi_{x,j}:=  |x^j,1\>\<x^j,1|- |x,1\>\<x,1|.
\]

\bigskip

In both scenarios, we have $O_{f_x} = \Pi_{x,j} + \Xi_{x,j}$ and $I_\regQ = \Pi_{x,j} + \Xi^2_{x,j}$. Moreover, $\Pi_{x,j}$ commutes with $\sigma_P^{(j)}$ for all $P\in\Pauli$, while $\Xi_{x,j}$ commutes with $\sigma_I^{(j)}$ and $\sigma_Z^{(j)}$, but anticommutes with $\sigma_X^{(j)}$ and $\sigma_Y^{(j)}$.%
\footnote{For this statement to hold in both scenarios is exactly the reason why we consider the phase oracle.}
Hence, we have that
\begin{equation}
\label{eq:GppViaPiXi}
G_{P,P'}^{x,j} = d_{P,r} d_{P',r} \sigma_P^{(j)}\sigma_{P'}^{(j)}
(\Pi_{x,j}+\mathfrak{c}_{Z,P'}\Xi_{x,j}),
\end{equation}
where $\mathfrak{c}_{Z,P'}=1$ if $P'\in\{I,Z\}$ and $\mathfrak{c}_{Z,P'}=-1$ if $P'\in\{X,Y\}$.%
\footnote{We can think of $\mathfrak{c}_{Z,P'}$ as the group-theoretic commutator $\sigma_Z\sigma_{P'}\sigma_Z\sigma_{P'}$.}

\subsection{A useful Kraus representation of the noisy oracle}

While we have already obtained one Kraus representation of $\ccO_{f,j,p}$,
not all Kraus representations are equally useful for our purposes. Therefore, let us present another Kraus representation of $\ccO_{f,j,p}$, which will turn out to be more useful.

It is known that two Kraus representations $\{G_i\}_i$ and $\{K_{i'}\}_{i'}$ correspond to the same CPTP map if and only if there exist coefficients $u_{i,i'}$ such that both $K_{i'}=\sum_i u_{i,i'}G_i$ for all $i'$ and the matrix $(u_{i,i'})_{i,i'}$ is unitary (see, for example, \cite{NielsenChuang,watrous2018}). 
We will refer to $\{K_{i'}\}_{i'}$ as a \emph{unitary combination} of $\{G_i\}_i$.

Let us define coefficients $a_P,b_P,c_P$ for $P\in\Pauli$ as follows:
\begin{subequations}
\label{eq:abcCoeffs}
\begin{alignat}{3}
& a_I := \frac{\sqrt{4-6r+3r^2}}{2}, &\qquad& b_I := \frac{(2-r)(1-r)}{\sqrt{4-6r+3r^2}}, &\qquad& c_I:= \frac{r\sqrt{8-12r+5r^2}}{2\sqrt{4-6r+3r^2}}, \\
& a_X := \frac{\sqrt{r(2-r)}}{2}, && b_X := 0 , && c_X:= \frac{\sqrt{r(2-r)}}{2}, \\
& a_Y := \frac{\sqrt{r(2-r)}}{2}, && b_Y := 0 , && c_Y:= \frac{\sqrt{r(2-r)}}{2}, \\
& a_Z := \frac{\sqrt{r(2-r)}}{2}, && b_Z :=  \frac{(1-r)\sqrt{r}}{\sqrt{2-r}} , && c_Z:=  \frac{r\sqrt{4-3r}}{2\sqrt{2-r}}.
\end{alignat}
\end{subequations}
In turn, let us define operators
\[
K_{0,P}^{x,j} := \sigma_P^{(j)}\big( a_P \Pi_{x,j} + b_P \Xi_{x,j}\big)
\qqAnd
K_{1,P}^{x,j} :=  \sigma_P^{(j)} c_P \Xi_{x,j}.
\]
We claim that the set of these operators form a Kraus representation of the noisy oracle call $\ccO_{f_x,j,r}$.

\begin{lem}
\label{lem:KrausAlt}
For any operator $\rho$ on the space corresponding to the register $\regQ$, we have
\[
\ccO_{f_x,j,r} (\rho) = \sum_{P\in\Pauli}\sum_{\beta\in\{0,1\}} 
K^{x,j}_{\beta,P} \rho (K^{x,j}_{\beta,P})^*.
\]
\end{lem}

\begin{proof}
Table~\ref{tab:KGKraus} shows the coefficients that express the set of operators $\{K^{x,j}_{\beta,P}\}_{\beta,P}$ as a unitary combination of the set of Kraus operators $\{G^{x,j}_{P',P''}\}_{P',P''}$. Since the space spanned by $\{G^{x,j}_{P',P''}\}_{P',P''}$ is 8-dimensional, a half of the linear combinations form $\{K^{x,j}_{\beta,P}\}_{\beta,P}$, the other half are all $0$ matrices. Note that, technically, we would have to provide a $16\times 16$ unitary matrix, but we can think of the four $4\times4$ unitary matrices given by Table~\ref{tab:KGKraus} to form one $16\times16$ block-diagonal matrix.
In other words, we set all the coefficients corresponding to pairs $(G^{x,j}_{P',P''}, K^{x,j}_{\beta,P})$ with $\sigma_P\not\propto \sigma_{P'}\sigma_{P''}$ to $0$, and do no display them in Table~\ref{tab:KGKraus}.
\end{proof}

\begin{table}[!h]
\begin{tabular}{l|cccc}
& $K_{1,I}^{x,j}$ & $K_{0,I}^{x,j}$ & $0$ & $0$ \\ \hline
$G_{I,I}^{x,j}=\frac{4-3r}{4}(\Pi_{x,j}+\Xi_{x,j})$ & $\frac{r(4-3r)}{2\sqrt{4-6r+3r^2}\sqrt{8-12r+5r^2}}$ 
& $\frac{4-3r}{2\sqrt{4-6r+3r^2}}$ & $\frac{r}{\sqrt{2(8-12r+5r^2)}}$ & 0 \\
$G_{X,X}^{x,j}= \frac{r}{4} (\Pi_{x,j}-\Xi_{x,j})$ &  $-\frac{\sqrt{8-12r+5r^2}}{2\sqrt{4-6r+3r^2}}$ & $\frac{r}{2\sqrt{4-6r+3r^2}}$ & 0 & $\frac1{\sqrt 2}$ \\
$G_{Y,Y}^{x,j}= \frac{r}{4} (\Pi_{x,j}-\Xi_{x,j})$ & $-\frac{\sqrt{8-12r+5r^2}}{2\sqrt{4-6r+3r^2}}$ & $\frac{r}{2\sqrt{4-6r+3r^2}}$ & 0 & $-\frac1{\sqrt 2}$ \\
$G_{Z,Z}^{x,j}= \frac{r}{4} (\Pi_{x,j}+\Xi_{x,j})$ & $\frac{r^2}{2\sqrt{4-6r+3r^2}\sqrt{8-12r+5r^2}}$ 
& $\frac{r}{2\sqrt{4-6r+3r^2}}$ & $-\frac{4-3r}{\sqrt{2(8-12r+5r^2)}}$ & 0
\end{tabular}

\bigskip

\begin{tabular}{l|cccc}
& $K_{1,X}^{x,j}$ & $K_{0,X}^{x,j}$ & $0$ & $0$
 \\ \hline %
$G_{I,X}^{x,j}=\frac{\sqrt{r(4-3r)}}{4} \sigma_X^{(j)} (\Pi_{x,j}-\Xi_{x,j}),$ & $-\frac{\sqrt{4-3r}}{2\sqrt{2-r}}$ & $\frac{\sqrt{4-3r}}{2\sqrt{2-r}}$ &
 $\frac{\sqrt{r}}{\sqrt{2(2-r)}}$ & $0$  
 \\ %
$G_{X,I}^{x,j}= \frac{\sqrt{r(4-3r)}}{4} \sigma_X^{(j)} (\Pi_{x,j}+\Xi_{x,j})$ & $\frac{\sqrt{4-3r}}{2\sqrt{2-r}}$ & $\frac{\sqrt{4-3r}}{2\sqrt{2-r}}$ &
 $0$ & $\frac{\sqrt{r}}{\sqrt{2(2-r)}}$ 
 \\ %
 $G_{Y,Z}^{x,j}= i  \frac{r}{4} \sigma_X^{(j)} (\Pi_{x,j}+\Xi_{x,j})$ & $-\frac{i\sqrt{r}}{2\sqrt{2-r}}$ & $-\frac{i\sqrt{r}}{2\sqrt{2-r}}$ & $0$ &
 $\frac{i\sqrt{4-3r}}{\sqrt{2(2-r)}}$
 \\ %
$G_{Z,Y}^{x,j}= -i\frac{r}{4} \sigma_X^{(j)} (\Pi_{x,j}-\Xi_{x,j})$ & $-\frac{i\sqrt{r}}{2\sqrt{2-r}}$ & $\frac{i\sqrt{r}}{2\sqrt{2-r}}$ & $-\frac{i\sqrt{4-3r}}{\sqrt{2(2-r)}}$ & $0$  
\end{tabular}

\bigskip

\begin{tabular}{l|cccc}
& $K_{1,Y}^{x,j}$ & $K_{0,Y}^{x,j}$ & $0$ & $0$
 \\ \hline %
$G_{I,Y}^{x,j}=\frac{\sqrt{r(4-3r)}}{4} \sigma_Y^{(j)} (\Pi_{x,j}-\Xi_{x,j})$ & $-\frac{\sqrt{4-3r}}{2\sqrt{2-r}}$ & $\frac{\sqrt{4-3r}}{2\sqrt{2-r}}$ &
 $\frac{\sqrt{r}}{\sqrt{2(2-r)}}$ & $0$  
 \\ %
$G_{Y,I}^{x,j}= \frac{\sqrt{r(4-3r)}}{4} \sigma_Y^{(j)} (\Pi_{x,j}+\Xi_{x,j})$ & $\frac{\sqrt{4-3r}}{2\sqrt{2-r}}$ & $\frac{\sqrt{4-3r}}{2\sqrt{2-r}}$ &
 $0$ & $\frac{\sqrt{r}}{\sqrt{2(2-r)}}$ 
 \\ %
 $G_{Z,X}^{x,j}=  i  \frac{r}{4} \sigma_Y^{(j)} (\Pi_{x,j}-\Xi_{x,j})$ & $\frac{i\sqrt{r}}{2\sqrt{2-r}}$ & $-\frac{i\sqrt{r}}{2\sqrt{2-r}}$ &
 $\frac{i\sqrt{4-3r}}{\sqrt{2(2-r)}}$ & $0$ 
 \\ %
$G_{X,Z}^{x,j}= -i  \frac{r}{4} \sigma_Y^{(j)} (\Pi_{x,j}+\Xi_{x,j})$ & $\frac{i\sqrt{r}}{2\sqrt{2-r}}$ & $\frac{i\sqrt{r}}{2\sqrt{2-r}}$ &
 $0$ & $-\frac{i\sqrt{4-3r}}{\sqrt{2(2-r)}}$ 
\end{tabular}

\bigskip

\begin{tabular}{l|cccc}
& $K_{1,Z}^{x,j}$ & $K_{0,Z}^{x,j}$ & $0$ & $0$ \\ \hline
$G_{Z,I}^{x,j}= \frac{\sqrt{r(4-3r)}}{4} \sigma_Z^{(j)} (\Pi_{x,j}+\Xi_{x,j})$ & $\frac{\sqrt{r}}{2\sqrt{2-r}}$ & $\frac{\sqrt{4-3r}}{2\sqrt{2-r}}$ &
 $\frac{1}{\sqrt{2}}$ & 0 \\
$G_{I,Z}^{x,j}= \frac{\sqrt{r(4-3r)}}{4} \sigma_Z^{(j)} (\Pi_{x,j}+\Xi_{x,j})$ & $\frac{\sqrt{r}}{2\sqrt{2-r}}$ & $\frac{\sqrt{4-3r}}{2\sqrt{2-r}}$ &
 $-\frac{1}{\sqrt{2}}$ & 0 \\
$G_{X,Y}^{x,j}= i\frac{r}{4} \sigma_Z^{(j)} (\Pi_{x,j}-\Xi_{x,j})$ &  $\frac{i\sqrt{4-3r}}{2\sqrt{2-r}}$ & $-\frac{i\sqrt{r}}{2\sqrt{2-r}}$ &
 0 & $\frac1{\sqrt 2}$ \\
$G_{Y,X}^{x,j}= -i\frac{r}{4} \sigma_Z^{(j)} (\Pi_{x,j}-\Xi_{x,j})$ & $-\frac{i\sqrt{4-3r}}{2\sqrt{2-r}}$ & $\frac{i\sqrt{r}}{2\sqrt{2-r}}$ &
 0 & $\frac1{\sqrt 2}$
\end{tabular}
\captionsetup{font=small}
\captionsetup{width=0.9\textwidth}
\caption[my caption]{%
Unitary matrices relating Kraus operators $G_{P',P''}^{x,j}$ and $K_{\beta,P}^{x,j}$.  The expressions for $G_{P',P''}^{x,j}$ are from (\ref{eq:GppViaPiXi}).}
\label{tab:KGKraus}
\end{table}

Let us list some useful observations about the coefficients $a_P$, $b_P$, $c_P$.

\begin{clm}
\label{clm:aPbPcP}
We have that
\begin{itemize}
\item both $a_I$, $b_I$ are at most $1-\frac r2$;
\item all $a_X$, $a_Y$, $a_Z$, $b_Z$, $c_I$, $c_X$, $c_Y$, $c_Z$ are at most $\sqrt{r/2}$;
\item $\sum_{P\in\Pauli}a_P^2=1$, $\sum_{P\in\Pauli}b_P^2\le 1$, and $\sum_{P\in\Pauli}c_P^2\le 2r$.
\end{itemize}
\end{clm}

Claim~\ref{clm:aPbPcP} can be verified by a direct computation. Note that, however, the statements $\sum_{P\in\Pauli}a_P^2=1$ and $\sum_{P\in\Pauli}b_P^2\le 1$ are also a corollary of Lemma~\ref{lem:KrausAlt}: because $\{K^{x,j}_{\beta,P}\}_{\beta,P}$ form a Kraus representation of a CPTP map, we have
\[
I_\regQ = \sum_{P\in\Pauli}\sum_{\beta\in\{0,1\}} 
(K^{x,j}_{\beta,P})^*K^{x,j}_{\beta,P} 
=
\Pi_{x,j}^2\sum_{P\in\Pauli}a_P^2 
+ \Xi_{x,j}^2\sum_{P\in\Pauli}(b_P^2 + c_P^2).
\]

\section{Progress measure}
\label{sec:progMes}

\subsection{Record registers}
\label{sec:recReg}

As in the original article, the record register $\regR$ starts empty, the subregister $\regR_1$ is appended to it by the first oracle call, $\regR_2$ by the second oracle call, and so on.
Now, for every $t\in\{1,\ldots,\tau\}$, the register $\regR_t$ consists of three qubits, and we further subdivide $\regR_t$ in subregisters as $\regR_t=\regR''_t\regR'_t$, where $\regR''_t$ are two qubits corresponding to the set $\Pauli$ of the labels of Pauli operators and $\regR'_t$ is a qubit corresponding to $\{0,1\}$.

We define the $t$-th \emph{extended} noisy oracle call as a linear isometry
\[
O_{j,r} := \sum_{x\in[n]}|f_x\>\<f_x|_\regT \otimes \sum_{P\in\Pauli}\sum_{\beta\in\{0,1\}} 
( K^{x,j}_{\beta,P})_\regQ \otimes |P,\beta\>_{\regR_t};
\]
it acts as the identity on the workspace register $\regW$ and earlier record subregisters $\regR_1\ldots \regR_{t-1}$.
Just like in the original article, \cite[Section 3.2]{Rosmanis:2023:NoisyOracle}, $O_{j,r}$ is a purification of the noisy oracle call in the sense that, if we call $O_{j,r}$ with the truth register $\regFunc$ being in the state $|f_x\>$ and then discard the newly introduced record register $\regR_t$, then the resulting map is equal to $\ccO_{f_x,j,r}$.
However, unlike in the original article, we do not write $O_{j,r}$ as the product of the faultless oracle call and a purification of the noise, which is the main technical innovation of this addendum.
Also unlike in the original article, here we do not use flag bits that indicate the presence or absence of error, therefore the size of the workspace does not change throughout the extended computation.

As before, the initial state of the extended computation is $|\phi_0\>:=|\unif\>\otimes|\psi^0\>$, where $|\unif\>:=\sum_{x\in[n]}|f_x\>/\sqrt{n}$, the intermediate states just before the extended oracle calls and the final state are defined as $|\phi_t\>:=U_tO_{j,r}|\phi_{t-1}\>$, where $t\in\{1,\ldots,\tau\}$, and the success probability for a randomly chosen $f_x$ is $q_{succ}=\|\Pi_{succ}|\phi_\tau\>\|^2$, where $\Pi_{succ}:=\sum_{x\in[n]}|f_x,x\>\<f_x,x|$.
 
\subsection{Progress-defining subspaces}
\label{ssec:progressSubsp}

When defining the progress measure, the content of registers $\regR''_t$, which correspond to $\Pauli$, will be completely irrelevant, while, if any of $\regR'_t$ will be in state $|1\>$, we will intuitively assume that a noise has collapsed the computation to the unique correct solution.

For the record register $\regR=\regR_1\ldots \regR_{t}$ containing $t$ entries, let us define complementary projectors
\[
\Lambda^{\progA\progB}_t:= \bigotimes_{s=1}^t 
\big(I_{\regR''_s}\otimes |0\>\<0|_{\regR'_s} \big)
\qqAnd
\Lambda^{\progC}_t:= I_\regR - \Lambda^{\progA\progB}_t.
\]

\begin{defn}
Define the \emph{progress projectors} on registers $\regFunc\regR$ as
\begin{align*}
&
\Pi^{\progC}_{t} :=
 I_\regT\otimes \Lambda^\progC_t,
\\ &
\Pi^{\progB}_{t} := 
 (I_\regT-|\unif\>\<\unif|) \otimes \Lambda^{\progA\progB}_t,
\\ &
\Pi^{\progA}_{t} := 
 |\unif\>\<\unif| \otimes \Lambda^{\progA\progB}_t.
\end{align*}
We define the \emph{progress measure} as
\[
\Psi_t:=\|\Pi^\progC_{t}|\phi_t\>\|^2 + 40 \|\Pi^\progB_{t}|\phi_t\>\|^2.
\]
\end{defn}

\bigskip

The following lemma shows that the progress measure $\Psi_{\tau}$ is, essentially, an upper bound on the success probability, and it also bounds by how much a single call to the noisy oracle can increase the value of the progress measure. The main theorem then easily follows the lemma.

\begin{lem}
\label{lem:progEvol}
We have 
\begin{subequations}
\begin{align}
& \Psi_0 = 0, \label{lem:progEvol:A} \\
& q_{succ} \le \Psi_\tau +  \frac2{n}, \label{lem:progEvol:B} \\
& \Psi_{t+1}-\Psi_{t} \le \frac{2000}{nr}. \label{lem:progEvol:C}
\end{align}
\end{subequations}
\end{lem}

\begin{proof}[Proof of Theorem~\ref{thm:mainSingle} given Lemma~\ref{lem:progEvol}]
From Lemma~\ref{lem:progEvol} we see that the success probability $q_{succ}$ is at most 
\[
\frac{2}{n}+\frac{2000\tau}{nr} \le \frac{2000(\tau+1)}{nr}.
\]
Since we want this success probability to be at least $1-\epsilon$, the bound follows.
\end{proof}

The first claim of Lemma~\ref{lem:progEvol}, equality (\ref{lem:progEvol:A}), is trivial. Towards the second claim, inequality (\ref{lem:progEvol:B}), as in the original article, we observe that $\Pi_{succ}$ commutes with both $\Pi^\progC_\tau$ and $\Pi^\progB_\tau+\Pi^\progA_\tau$ and we can write 
\begin{align*}
q_{succ}  
\le  \|\Pi^\progC_\tau|\phi_\tau\>\|^2 + 
2\|\Pi^\progB_\tau|\phi_\tau\>\|^2
+ 2\|\Pi_{succ}\Pi^\progA_\tau\|^2
 \le  \Psi_\tau + 2\|\Pi_{succ}\Pi^\progA_\tau\|^2.
\end{align*}
The claim holds due to
\[
\|\Pi_{succ}\Pi^\progA_\tau\| 
= \max_{x\in[n]} \||f_x\>\<f_x|\unif\>\<\unif|\| = 1/\sqrt{n}.
\]
The proof of the last claim of Lemma~\ref{lem:progEvol} is much more involved than those of the former two, and we devote the reminder of this note for that proof.

Consider the last claim, inequality (\ref{lem:progEvol:C}), and note that $0\le\Psi_t\le40$. Therefore, if $n\le 11$, then $2000/(nr)\ge 2000/11 >40$ and the claim holds. Hence, it is left to prove (\ref{lem:progEvol:C}) only for $n\ge 12$.

\section{Change in the Progress Measure}

In this section, similarly as in the original article, we first decompose the space $\cH^{\progB}$ as $\cH^{\progB,\actv}\oplus\cH^{\progB,\pasv}$. We then provide Claims~\ref{clm:CFixed}--\ref{clm:targ:toProgC}, which regard how the calls to the extended noisy oracle can transfer the probability weight among subspaces $\cH^{\progA}$, $\cH^{\progB,\actv}$, $\cH^{\progB,\pasv}$, $\cH^{\progC}$. Using these claims we then conclude the proof of Lemma~\ref{lem:progEvol}. We leave the proofs of Claims~\ref{clm:CFixed}--\ref{clm:targ:toProgC} to the next section.

\subsection{Active and passive subspaces}
\label{ssec:ActAndPas}

Let us start by decomposing the space $\cH^{\progB}$ as $\cH^{\progB,\actv}\oplus\cH^{\progB,\pasv}$. We will perform this decomposition differently depending on whether the target qubit is noisy ($j=0$) or one of the index qubits is noisy ($j\in\{1,\ldots,\log n\}$).

\paragraph{Noisy target qubit, $j=0$.}

Let us define
\[
|\unif_{x}\>:=\frac{1}{\sqrt{n-1}}\sum_{x'\in[n]\setminus\{x\}}|f_{x'}\>,
\]
and note that 
$\big(|f_x\>\<f_x|+|\unif_{x}\>\<\unif_{x}|\big)|\unif\>=|\unif\>$.
We decompose $\check\Pi^\progB_t\otimes I_\regQ = \check\Pi^{\progB,\actv}_{t} + \check\Pi^{\progB,\pasv}_{t}$, where
\begin{align*}
& \check\Pi^{\progB,\actv}_{t} := \sum_{x\in[n]}\big(|\unif_{x}\>\<\unif_{x}|+|f_x\>\<f_x|-|\unif\>\<\unif|\big) \otimes |x\>\<x| \otimes I_\regQo \otimes \Lambda^{\progA\progB}_{t},  \\
& \check\Pi^{\progB,\pasv}_{t} := \sum_{x\in[n]}
\big(I_\regT-|\unif_{x}\>\<\unif_{x}|-|f_x\>\<f_x|\big) \otimes |x\>\<x| \otimes I_\regQo \otimes \Lambda^{\progA\progB}_{t}.
\end{align*}
We have chosen the above definition of $\check\Pi^{\progB,\actv}_{t}$ to match more closely to the corresponding definition of $\hat\Pi^{\progB,\actv}_{t}$ in the noisy index qubit scenario below. However, if we want the definition of $\check\Pi^{\progB,\actv}_{t}$ to resemble more the definition used in the original article, we can observe that
\begin{align*}
\check\Pi^{\progB,\actv}_{t} = \sum_{x\in[n]}|\tilde f_x\>\<\tilde f_x| \otimes |x\>\<x| \otimes \Lambda^{\progA\progB}_{t}
\end{align*}
where $|\tilde{f}_{x}\>$ is an \emph{approximation} of $|f_x\>$ defined as  
\begin{align*}
|\tilde{f}_{x}\> := \,&
 \sqrt{\frac{n-1}{n}}|f_x\> - \frac1{\sqrt{n}}|\unif_x\>
=
  \frac{\sqrt{n}|f_x\>-|\unif\>}{\sqrt{n-1}}.
\end{align*}
Note that $|\tilde{f}_x\>$ and $|\unif\>$ are orthogonal.

\paragraph{Noisy index qubit, $j\in\{1,\ldots,\log n\}$.}

Recall that $x^j\in[n]$ is such that the binary representation of $x^j$ is obtained by flipping the $j$-th bit in the binary representation of $x$. Let us define
\[
|\unif_{x,x^j}\>:=\frac{1}{\sqrt{n-2}}\sum_{x'\in[n]\setminus\{x,x^j\}}|f_{x'}\>,
\]
and note that
$
\big(|f_x\>\<f_x|+|f_{x^j}\>\<f_{x^j}|+|\unif_{x,x^j}\>\<\unif_{x,x^j}|\big)|\unif\>=|\unif\>$.
We define orthogonal projectors
\begin{align*}
 \hat\Pi^{\progB,\actv}_{t} := & \sum_{x\in[n]}\big(|\unif_{x,x^j}\>\<\unif_{x,x^j}|+|f_x\>\<f_x|+|f_{x^j}\>\<f_{x^j}|-|\unif\>\<\unif|\big) \otimes |x,1\>\<x,1| \otimes \Lambda^{\progA\progB}_{t},  \\
 \hat\Pi^{\progB,\pasv}_{t} := & \sum_{x\in[n]}
\big(I_\regT-|\unif_{x,x^j}\>\<\unif_{x,x^j}|-|f_x\>\<f_x|-|f_{x^j}\>\<f_{x^j}|\big) \otimes |x,1\>\<x,1| \otimes \Lambda^{\progA\progB}_{t}
\\ &
+\big(I_\regT-|\unif\>\<\unif|\big)\otimes I_\regQi\otimes|0\>\<0|_\regQo \otimes \Lambda^{\progA\progB}_{t},
\end{align*}
which give us a decomposition
$\hat\Pi^\progB_t\otimes I_\regQ = \hat\Pi^{\progB,\actv}_{t} + \hat\Pi^{\progB,\pasv}_{t}$.

\paragraph{Invariance of the noisy qubit.}

Note that, in both scenarios, both $\Pi^{\progB,\actv}_{t}$ and $\Pi^{\progB,\pasv}_{t}$ act as the identity on the noisy query qubit. Hence, the Pauli operators $\sigma_I^{(j)},\sigma_X^{(j)},\sigma_Y^{(j)},\sigma_Z^{(j)}$ acting on that very qubit commute with both $\Pi^{\progB,\actv}_{t}$ and $\Pi^{\progB,\pasv}_{t}$.

\subsection{Transitions among progress-defining subspaces}
\label{ssec:transitions}

For both the target qubit noise and the index qubit noise, we claim the following.
First, we claim that some subspaces remain invariant under oracle calls.

\begin{clm}
\label{clm:CFixed}
We have $O_{j,r} \Pi^{\progC}_{t} =  \Pi^{\progC}_{t+1} O_{j,r} \Pi^{\progC}_{t}$, and the images of 
  $\Pi^{\progC}_{t+1} O_{j,r} (\Pi^{\progB}_{t}+\Pi^{\progA}_{t})$
 and
  $\Pi^{\progC}_{t+1} O_{j,r} \Pi^{\progC}_{t}$
   are orthogonal.
\end{clm}

\begin{clm}
\label{clm:OQonPas}
We have $O_{j,r} \Pi^{\progB,\pasv}_{t} =  \Pi^{\progB,\pasv}_{t+1} O_{j,r}$.
\end{clm}

\noindent
And then, informally speaking, we bound the transitions among various subspaces.

\begin{clm}
 \label{clm:targ:ProgAtoProgB}
We have
$
\| \Pi^{\progB,\actv}_{t+1} O_{j,r} \Pi^{\progA}_{t} \|  \le
2/\sqrt{n}$, and, assuming $n\ge 12$,  we also have $
\| \Pi^{\progB,\actv}_{t+1} O_{j,r} \Pi^{\progB,\actv}_{t} \|  \le 1 - r/9$.
\end{clm}

\begin{clm}
\label{clm:targ:toProgC}
 We have 
 $\| \Pi^{\progC}_{t+1} O_{j,r} \Pi^{\progB,\actv}_{t}\|^2 \le 2r$
  and 
$\| \Pi^{\progC}_{t+1} O_{j,r} \Pi^{\progA}_{t}\|^2\le 4r/n$.
 \end{clm}

We prove these claims in Section~\ref{sec:ProofsOfClaims}. Before we do that, let us use them to conclude the proof of Lemma~\ref{lem:progEvol} and, in turn, Theorem~\ref{thm:mainSingle}.

\begin{proof}[Proof of Lemma~\ref{lem:progEvol} given Claims~\ref{clm:CFixed}--\ref{clm:targ:toProgC}]
Recall that $|\phi_{t+1}\> = U_{t+1}O_{j,r}|\phi_t\>$ and that $U_{t+1}$ commutes with both $\Pi^\progC_{t+1}$ and $\Pi^\progB_{t+1}$. Hence we can express 
$\Psi_{t+1}$ as 
\begin{align*}
\Psi_{t+1}  =\, &
\|\Pi^\progC_{t+1} U_{t+1}O_{j,r}|\phi_t\>\|^2+40\|\Pi^\progB_{t+1} U_{t+1}O_{j,r} |\phi_t\>\|^2
\\
=\, &
\|\Pi^\progC_{t+1} O_{j,r}|\phi_t\>\|^2
+40\|\Pi^{\progB,\actv}_{t+1} O_{j,r} |\phi_t\>\|^2
+40\|\Pi^{\progB,\pasv}_{t+1} O_{j,r} |\phi_t\>\|^2
\\
=\, &
\|\Pi^\progC_{t+1} O_{j,r} \Pi^{\progC}_{t} |\phi_t\>\|^2
+ \|\Pi^\progC_{t+1} O_{j,r}(\Pi^{\progB,\actv}_{t}+\Pi^{\progA}_{t})|\phi_t\>\|^2
\\ & \qquad
+40\|\Pi^{\progB,\actv}_{t+1} O_{j,r} (\Pi^{\progB,\actv}_{t}+\Pi^{\progA}_{t}) |\phi_t\>\|^2
+40\|\Pi^{\progB,\pasv}_{t} |\phi_t\>\|^2
\\
=\, &
\Psi_t
+ \|\Pi^\progC_{t+1} O_{j,r}(\Pi^{\progB,\actv}_{t}+\Pi^{\progA}_{t})|\phi_t\>\|^2
-40 \|\Pi^{\progB,\actv}_{t} |\phi_t\>\|^2
\\ & \qquad
+40 \|\Pi^{\progB,\actv}_{t+1} O_{j,r} (\Pi^{\progB,\actv}_{t}+\Pi^{\progA}_{t}) |\phi_t\>\|^2,
  \end{align*}
  where the third equality is because $\Pi^\progC_{t+1} O_{j,r} \Pi^{\progC}_{t}$ and $\Pi^\progC_{t+1} O_{j,r} (\Pi^{\progB,\actv}_{t}+ \Pi^{\progA}_{t})$ have orthogonal images (Claim~\ref{clm:CFixed}) and because the passive space remains invariant under the extended noisy oracle call (Claim~\ref{clm:OQonPas}); the fourth equality uses Claim~\ref{clm:CFixed} to get
$
  \|\Pi^\progC_{t+1} O_{j,r} \Pi^{\progC}_{t} |\phi_t\>\|=\|\Pi^{\progC}_{t} |\phi_t\>\|.
  $

The above chain of equalities allows us to express $\Psi_{t+1}-\Psi_t$ using three norms. We bound two of them using the triangle inequality and Claims~\ref{clm:targ:ProgAtoProgB} and~\ref{clm:targ:toProgC}.
  First,   we have
\[
  \|\Pi^\progC_{t+1} O_{j,r}(\Pi^{\progB,\actv}_{t}+\Pi^{\progA}_{t})|\phi_t\>\|
\le
  \|\Pi^\progC_{t+1} O_{j,r} \Pi^{\progB,\actv}_{t}\|
  \!\cdot\!  \|\Pi^{\progB,\actv}_{t}|\phi_t\>\|
  +   \|\Pi^\progC_{t+1} O_{j,r} \Pi^{\progA}_{t}\|,
\]
and, in turn, taking its square, we get 
\begin{align*}
  \|\Pi^\progC_{t+1} O_{j,r}(\Pi^{\progB,\actv}_{t}+\Pi^{\progA}_{t})|\phi_t\>\|^2
& \le
  2\|\Pi^\progC_{t+1} O_{j,r} \Pi^{\progB,\actv}_{t}\|^2
  \!\cdot\!  \|\Pi^{\progB,\actv}_{t}|\phi_t\>\|^2
  +  2 \|\Pi^\progC_{t+1} O_{j,r} \Pi^{\progA}_{t}\|^2
\\ & \le
4r
  \!\cdot\!  \|\Pi^{\progB,\actv}_{t}|\phi_t\>\|^2
  +  8r/n,
\end{align*}
where the latter inequality is due to Claim~\ref{clm:targ:toProgC}.
Similarly, for $n\ge 12$, we have
  \begin{align*}
  \|\Pi^{\progB,\actv}_{t+1} O_{j,r} (\Pi^{\progB,\actv}_{t}+\Pi^{\progA}_{t}) |\phi_t\>\|
 & \le
 \|\Pi^{\progB,\actv}_{t+1} O_{j,r} \Pi^{\progB,\actv}_{t}\|\cdot \|\Pi^{\progB,\actv}_{t} |\phi_t\>\| +   \|\Pi^{\progB,\actv}_{t+1} O_{j,r} \Pi^{\progA}_{t}\|
 \\ & \le
\Big(1-\frac{r}{9}\Big)\|\Pi^{\progB,\actv}_{t} |\phi_t\>\| +   2/\sqrt{n},
  \end{align*}
where the latter inequality is due to Claim~\ref{clm:targ:ProgAtoProgB}.

Using the two inequalities above and observing that $(1-r/9)^2\le 1-2r/9 +r/81 \le 1-r/5$, we get 
\begin{align*}
\Psi_{t+1}  - \Psi_t 
\le\, &
4r \!\cdot\!  \|\Pi^{\progB,\actv}_{t}|\phi_t\>\|^2 + 8r/n
-40 \|\Pi^{\progB,\actv}_{t} |\phi_t\>\|^2
\\ &
+40
  \bigg(
\Big(1-\frac{r}{9}\Big)\|\Pi^{\progB,\actv}_{t} |\phi_t\>\| +   2/\sqrt{n}
 \bigg)^2
\\
\le\, &
4r
  \!\cdot\!  \|\Pi^{\progB,\actv}_{t}|\phi_t\>\|^2
  +  \frac{8}{nr}
  -40 \|\Pi^{\progB,\actv}_{t} |\phi_t\>\|^2
\\ &
+40
  \bigg(
\Big(1-\frac{r}{5}\Big)\|\Pi^{\progB,\actv}_{t} |\phi_t\>\|^2 + \frac{4}{\sqrt n}\|\Pi^{\progB,\actv}_{t} |\phi_t\>\|  + \frac{4}{nr} \bigg)
\\ = \, &
-\bigg(2\sqrt{r}\|\Pi^{\progB,\actv}_{t} |\phi_t\>\|
- \frac{40}{\sqrt{nr}}
\bigg)^2  
+ \frac{1768}{nr}
\\ \le \, & \frac{2000}{nr}. \qedhere
  \end{align*}  
\end{proof}

\section{Proofs of Claims~\ref{clm:CFixed}--\ref{clm:targ:toProgC}}
\label{sec:ProofsOfClaims}

Here we restate and prove the four claims used in the proof of Lemma~\ref{lem:progEvol}.

\bigskip
\noindent
\textbf{Claim~\ref{clm:CFixed}} (restated)\textbf{.}
\emph{
We have $O_{j,r} \Pi^{\progC}_{t} =  \Pi^{\progC}_{t+1} O_{j,r} \Pi^{\progC}_{t}$, and the images of 
  $\Pi^{\progC}_{t+1} O_{j,r} (\Pi^{\progB}_{t}+\Pi^{\progA}_{t})$
 and
  $\Pi^{\progC}_{t+1} O_{j,r} \Pi^{\progC}_{t}$
   are orthogonal.
} 

\begin{proof}
Essentially, the claim follows from definitions of $\Lambda^{\progA\progB}_t$ and $\Lambda^{\progC}_t$ and the observation that the extended noisy oracle call only appends new entries to the record register, never modifying the existing ones. This means that, if there was $s\in\{1,2,\ldots,t\}$ with the register $\regR'_s$ in state $|1\>$ already before the oracle call, then that register will remain in state $|1\>$ after the oracle call too, implying the former part of the claim.
Similarly we show the latter part: the content of the record registers $\regR_1\regR_2\ldots\regR_t$ of the states in the images of $\Pi^{\progB}_{t}+\Pi^{\progA}_{t}$ and $\Pi^{\progC}_{t}$ are orthogonal, and they will remain unchanged after applying $\Pi^{\progC}_{t+1} O_{j,r} $.
\end{proof}

\noindent
\textbf{Claim~\ref{clm:OQonPas}} (restated)\textbf{.}
\emph{
We have $O_{j,r} \Pi^{\progB,\pasv}_{t} =  \Pi^{\progB,\pasv}_{t+1} O_{j,r}$.
} 

\begin{proof} 
First, note that $\Pi_{t+1}^{\progB,\pasv}=\Pi_{t}^{\progB,\pasv}\otimes\sum_{P\in\Pauli}|P,0\>\<P,0|$, 
and, second, recall the expression for $O_{j,r}$ from Section~\ref{sec:recReg}, 
which, using the expressions for $K^{x}_{1,P}$ and $K^{x}_{0,P}$, we can rewrite it as 
\begin{align*}
O_{j,r} = \, &
\sum_{P\in\Pauli} \bigg[
\sigma_P^{(j)}
\sum_{x\in[n]}|f_x\>\<f_x| \otimes c_P\Xi_{x,j}
\bigg]
\otimes |P,1\>
\\ & +
\sum_{P\in\Pauli} \bigg[
\sigma_P^{(j)}
\sum_{x\in[n]}|f_x\>\<f_x| \otimes (a_P\Pi_{x,j}+b_P\Xi_{x,j})
\bigg]
\otimes |P,0\>.
\end{align*}
It suffices to show that, for every $P\in\Pauli$, the product of the operator in the former square brackets and  $\Pi^{\progB,\pasv}_t$ is $0$, while the operator in the latter square brackets commutes with $\Pi^{\progB,\pasv}_t$. 

Even though the definitions of $\Xi_{x,j}$ and $\Pi^{\progB,\pasv}$ differ in the noisy target qubit and the noisy index qubit scenarios, in both scenarios we have
\[
(|f_x\>\<f_x|\otimes \Xi_{x,j}) \Pi^{\progB,\pasv} = 0 = (|f_x\>\<f_x|\otimes \Xi_{x,j}^2) \Pi^{\progB,\pasv}.
\]
This already satisfies the requirement on the former square bracket.

As for the latter square bracket, because  $\Pi^{\progB,\pasv}_t$ acts as the identity on the $j$-th query qubit and, thus, the Pauli operators $\sigma^{(j)}_I,\sigma^{(j)}_X,\sigma^{(j)}_Y,\sigma^{(j)}_Z$  commute with $\Pi^{\progB,\pasv}_t$, it is left to show that $\sum_{x\in[n]}|f_x\>\<f_x| \otimes \Pi_{x,j}$ commutes with $\Pi^{\progB,\pasv}_t$.
This follows from both of these operators being Hermitian and $\Pi_{x,j}=I_\regQ-\Xi_{x,j}^2$, namely,
\begin{multline*}
\Big(\sum_{x\in[n]}|f_x\>\<f_x| \otimes \Pi_{x,j}\Big)\Pi^{\progB,\pasv}_t
= \Big(\sum_{x\in[n]}|f_x\>\<f_x| \otimes \big(I_\regQ-\Xi^2_{x,j}\big)\Big)\Pi^{\progB,\pasv}_t
\\
= \Big(\sum_{x\in[n]}|f_x\>\<f_x| \otimes I_\regQ\Big)\Pi^{\progB,\pasv}_t
= \Pi^{\progB,\pasv}_t.
\qedhere
\end{multline*}
\end{proof}

\subsection{Getting to and remaining in subspace $\cH^{\progB,\actv}$}

This section is devoted to proving Claim~\ref{clm:targ:ProgAtoProgB}, and we leave the proof of Claim~\ref{clm:targ:toProgC} to Secion~\ref{sec:targ:toProgC}.
Let us define the ``Pauli-$\sigma_P$-free'' part of Kraus operators as
\[
F_{0,P}^{x,j} :=  a_P \Pi_{x,j} + b_P \Xi_{x,j}
\qqAnd
F_{1,P}^{x,j} :=   c_P \Xi_{x,j}
\]
so that we have $K_{\beta,P}^{x,j} = \sigma_P^{(j)} F_{\beta,P}^{x,j}$ for all $P\in\Pauli$ and $\beta\in\{0,1\}$. 
The operators $F_{\beta,P}^{x,j}$ are diagonal in the computational basis.

Claims~\ref{clm:targ:ProgAtoProgB} and \ref{clm:targ:toProgC} concern both $ O_{j,r}\Pi^{\progA}_t$ and $ O_{j,r}\Pi^{\progB,\pasv}_t$. To simultaneously make statements about both of these operators, let $\Pi^{\mathfrak{Lab}}_t\in\big\{\Pi^{\progA}_t,\Pi^{\progB,\actv}_t\big\}$, where $\progLab$ stands for ``Label''.

\bigskip
\noindent
\textbf{Claim~\ref{clm:targ:ProgAtoProgB}} (restated)\textbf{.}
\emph{
We have
$\| \Pi^{\progB,\actv}_{t+1} O_{j,r} \Pi^{\progA}_{t} \|  \le
2/\sqrt{n}$, and, assuming $n\ge 12$,  we also have $
\| \Pi^{\progB,\actv}_{t+1} O_{j,r} \Pi^{\progB,\actv}_{t} \|  \le 1 - r/9$.
}

\begin{proof} 
Because the states $|0,P\>_{\regR_{t+1}}$ introduced by $O_{j,r}$ are clearly orthogonal for distinct $P$ and because Pauli operators $\sigma_P^{(j)}$ commute with $\Pi^{\progB,\actv}$, we have
\begin{align*}
\| \Pi^{\progB,\actv}_{t+1}  O_{j,r} \Pi_t^{\progLab}\|^2
 & \le 
 \sum_{P\in\Pauli}\bigg\|
 \sum_{x'\in [n]}
\Pi^{\progB,\actv}_{t} 
  \big( |f_{x'}\>\<f_{x'}| \otimes 
  F_{0,P}^{x',j} \big)
  \Pi_t^{\progLab}\bigg\|^2.
\end{align*}
Note the change of the subscript from $t+1$ in $\Pi^{\progB,\actv}_{t+1}$ to $t$ in $\Pi^{\progB,\actv}_{t}$, which is for the purely technical reason that we are dropping the register $\regR_{t+1}$.

For any given $x\in[n]$, we note that $|x\>\<x|_\regQi$ commutes with all $\Pi^{\progB,\actv}_{t}$, $F_{0,P}^{x',j}$, and $\Pi_t^{\progLab}$.
Hence, due to symmetry, we have 
\[
\bigg\|
 \sum_{x'\in [n]}
\Pi^{\progB,\actv}_{t} 
  \big( |f_{x'}\>\<f_{x'}| \otimes 
  F_{0,P}^{x',j} \big)
  \Pi_t^{\progLab}\bigg\|
  =
\bigg\|
 \sum_{x'\in [n]}
 \big(\<x|\otimes I\big)
\Pi^{\progB,\actv}_{t} 
  \big( |f_{x'}\>\<f_{x'}| \otimes 
  F_{0,P}^{x',j} \big)
  \Pi_t^{\progLab}
  \big(|x\>\otimes I\big) \bigg\|.
\]

\paragraph{Noisy target qubit.}
Note that, first, we have 
\[
\check\Pi^{\progB,\actv}_{t} \big(|x\>\otimes I\big)= |\tilde f_x,x\>\<\tilde f_x|\otimes I_\regQo \otimes \Lambda^{\progA\progB}_{t}
\qqAnd
\check\Pi^{\progA}_{t} \big(|x\>\otimes I\big) = |\unif,x\>\<\unif|\otimes I_\regQo \otimes \Lambda^{\progA\progB}_{t}
\]
and, second, we have
\[
\check F_{0,P}^{x,j} \big(|x\>\otimes I_{\regQo}\big)= b_P |x\>\otimes \sigma_Z
\qqAnd
\check F_{0,P}^{x',j} \big(|x\>\otimes I_{\regQo}\big)= a_P |x\>\otimes \sigma_I
\]
for $x'\ne x$. Since $\sum_{x'\in [n]\setminus\{x\}}  |f_{x'}\>\<f_{x'}| = I_\regT - |f_x\>\<f_x|$, we have
\begin{align*}
\| \check\Pi^{\progB,\actv}_{t+1}  O_{j,r} \check\Pi_t^{\progA}\|^2
 & \le 
 \sum_{P\in\Pauli}\bigg\|
\big(\<\tilde{f}_{x}|\otimes I\big) 
  \Big( b_P |f_x\>\<f_x|\otimes \sigma_Z 
  \\ & \hspace{100pt}
  +  a_P \big(I_\regT- |f_{x}\>\<f_{x}|\big) \otimes \sigma_I
\Big)\big(|\unif\>\otimes I\big)\bigg\|^2
\\  & =
 \sum_{P\in\Pauli}\bigg\|
\<\tilde{f}_{x}| 
  \Big( b_P |f_x\>\<f_x|\otimes \sigma_Z -  a_P  |f_{x}\>\<f_{x}| \otimes \sigma_I
\Big)|\unif\>\bigg\|^2
\\  & =
 \sum_{P\in\Pauli}
 |\<\tilde{f}_{x}|f_x\>|^2  \cdot
  |\<f_{x}|\unif\>|^2  \cdot
  \| b_P \sigma_Z -  a_P \sigma_I\|^2
\\  & =
 \sum_{P\in\Pauli}
 \frac{n-1}n \cdot  \frac1n  \cdot  (a_P+b_P)^2  
\\  & \le
\frac2n
 \sum_{P\in\Pauli}
 (a_P^2+b_P^2) \le \frac4n,
\end{align*}
where the last inequality is due to Claim~\ref{clm:aPbPcP}.
Similarly, because $b_X=0$ and $b_Y=0$, we have
\begin{align*}
\| \check\Pi^{\progB,\actv}_{t+1}  O_{j,r} \check\Pi_t^{\progB,\actv}\|^2
 & \le 
 \sum_{P\in\Pauli}\bigg\|
\big(\<\tilde{f}_{x}|\otimes I\big) 
  \Big( b_P |f_x\>\<f_x|\otimes \sigma_Z 
    \\ & \hspace{100pt} +  a_P 
  \big(I_\regT- |f_{x}\>\<f_{x}|\big)  \otimes \sigma_I
\Big)\big(|\tilde{f}_x\>\otimes I\big)\bigg\|^2
\\ & \le 
 \sum_{P\in\{I,Z\}}\Big\|
 b_P |f_x\>\<f_x|\otimes \sigma_Z +  a_P 
 \big(I_\regT- |f_{x}\>\<f_{x}|\big) \otimes \sigma_I
\Big\|^2
\\ & \quad +
 \sum_{P\in\{X,Y\}} a_P^2
 \bigg\|
\<\tilde{f}_{x}| 
 \big(I_\regT- |f_{x}\>\<f_{x}|\big) |\tilde{f}_x\> \bigg\|^2
\\ & =
 \sum_{P\in\{I,Z\}} \max\{a_P^2,b_P^2\} +
 \sum_{P\in\{X,Y\}} a_P^2
 \Big(1 -
\frac{n-1}{n}
\Big)^2.
  \end{align*}
 The following claim, which we will also use for the noisy index qubit scenario, lets us conclude that
  $\| \check\Pi^{\progB,\actv}_{t+1} O_{j,r} \check\Pi^{\progB,\actv}_{t} \|  \le 1 - r/9$ for $n\ge12$.
  
  \begin{clm}
  \label{clm:aPbPcP:cor}
Assuming $n\ge 12$, we have 
\[
\sum_{P\in\{I,Z\}} \max\{a_P^2,b_P^2\} +
 \frac{4}{n^2}\sum_{P\in\{X,Y\}} a_P^2 
 \le \Big(1-\frac{r}{9}\Big)^2.
\]
  \end{clm}
  
  \begin{proof}
  Using, first, Claim~\ref{clm:aPbPcP} and, then, $r^2\le r$,  we get 
   \begin{multline*}
   \sum_{P\in\{I,Z\}} \max\{a_P^2,b_P^2\} +
 \frac{4}{n^2}\sum_{P\in\{X,Y\}} a_P^2 
 \le 
\Big(1-\frac r2\Big)^2  +\frac r2 + \frac{4}{n^2} \Big(\frac r2+\frac r2\Big) 
\\ \le 
1-r+\frac r4   +\frac r2 + \frac{4}{12^2}r 
 = 
1-\frac {2r}9
\le \Big(1-\frac{r}{9}\Big)^2.
\qedhere
 \end{multline*}

 \end{proof}

\paragraph{Noisy index qubit.}

To begin with, because $\hat\Pi_t^{\progB,\actv}(I\otimes |0\>_{\regQo})=0$, we are only interested in the case when the target qubit is in state $|1\>$. 
Note that, first, we have 
\[
\hat\Pi^{\progB,\actv}_{t} \big(|x,1\>\otimes I\big)= 
\big(|\unif_{x,x^j}\>\<\unif_{x,x^j}|+|f_x\>\<f_x|+|f_{x^j}\>\<f_{x^j}|-|\unif\>\<\unif|\big) \otimes |x,1\>\otimes \Lambda^{\progA\progB}_{t},
\]
\[
\hat\Pi^{\progA}_{t} \big(|x,1\>\otimes I\big) = |\unif,x,1\>\<\unif| \otimes \Lambda^{\progA\progB}_{t},
\]
second, we have
\[
\hat F_{0,P}^{x,j} |x,1\> = - b_P |x,1\>,
\qquad
\hat F_{0,P}^{x^j,j} |x,1\> = b_P |x,1\>,
\qqAnd
\hat F_{0,P}^{x',j} |x,1\> = a_P |x,1\>
\]
for $x'\notin\{x,x^j\}$, and, third, we have
\[
\sum_{x'\in[n]\setminus\{x,x^j\}} |f_{x'}\>\<f_{x'}| = I_\regT-|f_{x}\>\<f_{x}|-|f_{x^j}\>\<f_{x^j}|.
\]
Using these observations, we get
\begin{multline*}
\| \hat\Pi^{\progB,\actv}_{t+1}  O_{j,r} \hat\Pi_t^{\progA}\|^2 
   \le 
 \sum_{P\in\Pauli}\bigg\|
 \big(|\unif_{x,x^j}\>\<\unif_{x,x^j}|+|f_x\>\<f_x|+|f_{x^j}\>\<f_{x^j}|-|\unif\>\<\unif|\big)
 \\ 
  \cdot\Big(
  -b_P |f_{x}\>\<f_{x}|+b_P |f_{x^j}\>\<f_{x^j}| +a_P\big(I_\regT-|f_{x}\>\<f_{x}|- |f_{x^j}\>\<f_{x^j}|\big) \Big)
 \cdot|u\>\<u|\bigg\|^2.
\end{multline*}
Note that the operator given by the first parenthesis is orthogonal to $|u\>\<u|$, therefore the $I_\regT$ term gets eliminated. Thus, when we multiply the operator under the norm by $-1$ and drop the term $\<u|$, we get
\begin{align*}
& 
 \big(|\unif_{x,x^j}\>\<\unif_{x,x^j}|+|f_x\>\<f_x|+|f_{x^j}\>\<f_{x^j}|-|\unif\>\<\unif|\big)
 \\ & \hspace{65pt}
  \cdot\big(
  b_P |f_{x}\>\<f_{x}|-b_P |f_{x^j}\>\<f_{x^j}| +a_P |f_{x}\>\<f_{x}|+ a_P |f_{x^j}\>\<f_{x^j}|\big) 
 \cdot|u\>
\\ & = 
\big(|f_x\>\<f_x|+|f_{x^j}\>\<f_{x^j}|-|\unif\>\<\unif|\big)
\\ & \hspace{65pt}
 \cdot
\big(
a_P |f_{x}\>
+ a_P |f_{x^j}\>
+  b_P |f_{x}\>
- b_P |f_{x^j}\>
\big)/\sqrt{n}
\\ & =
\frac1{\sqrt n} \Big( 
 ( a_P+b_P) |f_{x}\>
+ (a_P-b_P) |f_{x^j}\>
- 2 a_P \frac{|\unif\>}{\sqrt{n}}
\Big)
\\ & =
\frac1{\sqrt n} \bigg( 
 \Big( \frac{n-2}{n}a_P+b_P\Big) |f_{x}\>
+ \Big(\frac{n-2}{n} a_P-b_P\Big) |f_{x^j}\>
-2 a_P \frac{\sqrt{n-2}|\unif_{x,x^j}\>}{n}
\bigg).
\end{align*}
The norm of this vector is 
\begin{align*}
\frac1n  \bigg( 
\Big(\frac{n-2}{n} a_P+ b_P\Big)^2
+ \Big(\frac{n-2}{n} a_P- b_P\Big)^2
+ \Big(\frac{2\sqrt{n-2}}{n} a_P\Big)^2
\bigg)
 = 
 \frac2n \Big( 
\frac{n-2}{n} a_P^2
+  b_P^2
\Big),
\end{align*}
which is at most
$\frac2n(  a_P^2 +  b_P^2)$.
We conclude bounding $\| \hat\Pi^{\progB,\actv}_{t+1}  O_{j,r} \hat\Pi_t^{\progA}\|$ by recalling $\sum_{P\in\Pauli}a_P^2=1$ and $\sum_{P\in\Pauli}b_P^2\le 1$ from Claim~\ref{clm:aPbPcP}.

\bigskip

Now let us bound $\|\hat\Pi^{\progB,\actv}_{t+1} O_{j,r}\hat\Pi^{\progB,\actv}_t\|$.
Since $b_X = b_Y = 0$ and since $\sum_{x'\in[n]\setminus\{x,x^j\}} |f_{x'}\>\<f_{x'}|$ restricted to $\mathrm{span}\{|\unif_{x,x^j}\>,  |f_x\>, |f_{x^j}\>, |\unif\>\}$ equals $|\unif_{x,x^j}\>\<\unif_{x,x^j}|$, we have
\begin{align*}
\big\|\hat\Pi^{\progB,\actv}_{t+1} O_{j,r}\hat\Pi^{\progB,\actv}_t\big\|^2 \le & \,
\sum_{P\in\{I,Z\}} \Big(
\big\||\unif_{x,x^j}\>\<\unif_{x,x^j}|+|f_x\>\<f_x|+|f_{x^j}\>\<f_{x^j}|-|\unif\>\<\unif|\big\|
\\ & \hspace{45pt}
 \cdot
\Big\|
a_P |\unif_{x,x^j}\>\< \unif_{x,x^j}|
- b_P |f_{x}\>\<f_{x}|
+ b_P |f_{x^j}\>\<f_{x^j}|
\Big\|
\\ & \hspace{45pt}
 \cdot
\big\||\unif_{x,x^j}\>\<\unif_{x,x^j}|+|f_x\>\<f_x|+|f_{x^j}\>\<f_{x^j}|-|\unif\>\<\unif|\big\| \Big)^2
\\ & \, +
\sum_{P\in\{X,Y\}} a_P^2 \Big\| 
\big(|\unif_{x,x^j}\>\<\unif_{x,x^j}|+|f_x\>\<f_x|+|f_{x^j}\>\<f_{x^j}|-|\unif\>\<\unif|\big)
|\unif_{x,x^j}\>
\\ & \hspace{70pt}
 \cdot
\<\unif_{x,x^j}|
\big(|\unif_{x,x^j}\>\<\unif_{x,x^j}|+|f_x\>\<f_x|+|f_{x^j}\>\<f_{x^j}|-|\unif\>\<\unif|\big) \Big\|^2
\\ = & \,
\sum_{P\in\{I,Z\}} \max\{a_P^2,b_P^2\}
 + \sum_{P\in\{X,Y\}} a_P^2 \bigg\|  |\unif_{x,x^j}\>-\sqrt{\frac{n-2}{n}}|\unif\> \bigg\|^4.
\end{align*}
We have
\[
 \bigg\|
|\unif_{x,x^j}\>-\sqrt{\frac{n-2}{n}}|\unif\>
\bigg\|^2
 = 1 - \frac{n-2}n - \frac{n-2}n + \frac{n-2}n
 = \frac2n,
 \]
 and thus
\[
\big\|\hat\Pi^{\progB,\actv}_{t+1} O_{j,r}\hat\Pi^{\progB,\actv}_t\big\|^2 \le 
\sum_{P\in\{I,Z\}} \max\{a_P^2,b_P^2\}
 + \frac4{n^2}\sum_{P\in\{X,Y\}} a_P^2. 
\]
We again conclude the proof by Claim~\ref{clm:aPbPcP:cor}.
\end{proof}

\subsection{Getting to subspace $\cH^{\progC}$}
\label{sec:targ:toProgC}

\noindent
\textbf{Claim~\ref{clm:targ:toProgC}} (restated)\textbf{.}
\emph{
 We have 
 $\| \Pi^{\progC}_{t+1} O_{j,r} \Pi^{\progB,\actv}_{t}\|^2 \le 2r$
  and 
$\| \Pi^{\progC}_{t+1} O_{j,r} \Pi^{\progA}_{t}\|^2\le 4r/n$.
} 

\begin{proof}
As before, let  $\Pi_t^\progLab\in\{\Pi_t^\progA,\Pi_t^{\progB,\actv}\}$. 
Because the states $|P,1\>_{\regR_{t+1}}$ introduced by $O_{j,r}$ are clearly orthogonal for distinct $P$,  because $\Pi^{\progC}_{t+1}(I\otimes |P,1\>_{\regR_{t+1}})=I\otimes |P,1\>_{\regR_{t+1}}$, and because $\sigma_P$ is unitary,   we have
\begin{align*}
\| \Pi^{\progC}_{t+1} O_{j,r} \Pi^{\progLab}_{t}\big\|^2  
 & = 
 \bigg\|\Pi^{\progC}_{t+1}  \sum_{x\in[n]}
\Big(|f_x\>\<f_x|\otimes 
\sum_{P\in\Pauli} K_{1,P}^x \otimes |P,1\>\Big) \Pi^{\progLab}_{t}\bigg\|^2
\\ & \le  
\sum_{P\in\Pauli}
 \bigg\|  \sum_{x\in[n]}
\big(|f_x\>\<f_x|\otimes  \sigma_P^{(j)} c_P \Xi_{x,j} \big) \Pi^{\progLab}_{t}\bigg\|^2
\\ & = 
\bigg(
\sum_{P\in\Pauli} c_P^2
\bigg)
 \bigg\|  \sum_{x\in[n]} \big(|f_x\>\<f_x|\otimes  \Xi_{x,j} \big) \Pi^{\progLab}_{t}\bigg\|^2.
\end{align*}
Recall that $ \sum_{P\in\Pauli} c_P^2\le 2r$ by Claim~\ref{clm:aPbPcP}.

Note that $\big\| \sum_{x\in[n]} \big(|f_x\>\<f_x|\otimes  \Xi_{x,j} \big)\big\|=1$, already giving us 
$\big\| \Pi^{\progC}_{t+1} O_{j,r} \Pi^{\progB,\actv}_{t}  \big\|^2
\le 2r$.  
It is left to bound $\| \Pi^{\progC}_{t+1} O_{j,r} \Pi^{\progA}_{t}\|^2$, which we do separately for the two scenarios.

\paragraph{Noisy target qubit.}
 
Recall that $\check \Xi_{x,j} = |x\>\<x|\otimes \sigma_Z$ and note that $\|\sigma_Z\|=1$. Hence, we have 
\[
\bigg\| \sum_{x\in[n]} \Big(|f_x\>\<f_x|\otimes  \check\Xi_{x,j} \Big) \Pi^{\progA}_{t} \bigg\|
= \bigg\|  |f_x\>  \otimes |x\>\<x|\otimes\sigma_Z  \bigg\| / \sqrt{n} = 1/\sqrt{n},
\]
where we have used $\|\<f_x| \Pi^{\progA}_{t}\|=|\<f_x|\unif\>|=1/\sqrt{n}$.

\paragraph{Noisy index qubit.}

Recall that, for distinct $x,x'$, the operators $\hat \Xi_{x,j}$ and $\hat \Xi_{x',j}$ are mutually orthogonal, unless $x'=x^j$, for which $\hat \Xi_{x^j,j}=-\hat \Xi_{x,j}$.
We use $\<f_x|\unif\>=1/\sqrt n$ to get 
\[
\bigg\| \sum_{x\in[n]} \Big(|f_x\>\<f_x|\otimes  \hat\Xi_{x,j} \Big) \Pi^{\progA}_{t} \bigg\|
= \bigg\| \sum_{x\in\{x',x^{j'}\}} |f_x\>  \otimes \hat\Xi_{x,j}  \bigg\| / \sqrt{n} = \sqrt{2/n},
\]
where $x'$ can be arbitrary.
\end{proof}

\section*{Acknowledgements}

The author would like to thank Fran\c{c}ois Le Gall for fruitful and insightful discussions.
Part of the work was done while the author was at Graduate School of Mathematics, Nagoya University supported by JSPS KAKENHI Grant No.~JP20H05966 and MEXT Quantum Leap Flagship Program (MEXT Q-LEAP) Grant No.~JPMXS0120319794.
Part of the work was done while the author was hosted by the Centre for Quantum Technologies, the National University of Singapore, his friends, and and his family, and the author would like to thank all of them for their hospitality.

\appendix

\section{Negligent Oracle}
\label{app:negligent}

In this section, we show how to utilize the techniques introduced above to reprove the $\Omega(np)$-query lower bound by Regev and Schiff on the quantum query complexity of search with the negligent oracle \cite{regev:faultySearch}. Note that the techniques of the original article~\cite{Rosmanis:2023:NoisyOracle} were insufficient for this purpose because they gave the same bounds for the error-concealing and the error-signaling noisy oracles.

The negligent oracle call of noise rate $p$ is defined as a CPTP map on the query register $\regQ=\regQi\regQo$ that with probability $1-p$ acts as the faultless oracle call $O_f$ (given in (\ref{eq:faultless})) and with probability $p$ acts as the identity operation. 
We denote this CPTP map as $\widetilde\ccO_{f,p}$.

\begin{thm}
\label{thm:negligent}
Every algorithm that is given an access to the negligent oracle of noise rate $p$ and that solves the unstructured search problem in the worst case setting with probability at least $1-\epsilon$ must make at least $\big(\frac{np(1-\epsilon)}{28}-1\big)/(1-p)$ queries to the negligent oracle.
\end{thm}

To prove this lower bound, as in the main text, we only consider functions $f_x$, where $x\in[n]$ is the unique marked element. For such functions, $\widetilde\ccO_{f_x,p}$ has a Kraus representation 
\[
\widetilde\ccO_{f_x,p} \colon \rho \mapsto 
G_{x,0}\rho G_{x,0}^* + G_{x,1}\rho G_{x,1}^*
\]
with Kraus operators $G_{x,0}:=\sqrt{1-p}O_{f_x}$ and $G_{x,1}:=\sqrt{p}I_\regQ$.

Let $\{K_{x,0},K_{x,1}\}$ be a unitary combination of $\{G_{x,0},G_{x,1}\}$ given as
\begin{alignat*}{3}
& K_{x,0}&&:=\sqrt{1-p}G_{x,0} + \sqrt{p}G_{x,1} 
 = I_\regQ-2(1-p)|x,1\>\<x,1|,
\\
& K_{x,1}&&:=-\sqrt{p}G_{x,0} + \sqrt{1-p}G_{x,1} 
 = 2\sqrt{p(1-p)}|x,1\>\<x,1|,
\end{alignat*}
where we have used the fact that $O_{f_x}=I_\regQ-2|x,1\>\<x,1|$.
Therefore,  $\widetilde\ccO_{f_x,p}$ also has a Kraus representation 
\[
\widetilde\ccO_{f_x,p} \colon \rho \mapsto 
K_{x,0}\rho K_{x,0}^* + K_{x,1}\rho K_{x,1}^*,
\]
which we will be using from now on.

\subsection{Progress measure}

\paragraph{Record registers.}

Similarly as in Section~\ref{sec:recReg}, the record register $\regR$ starts empty, the subregister $\regR_1$ is appended to it by the first oracle call, $\regR_2$ by the second oracle call, and so on.
Now, however, each subregister $\regR_t$ consists of a single qubit.
We define the $t$-th \emph{extended} negligent oracle call as a linear isometry
\begin{align*}
\widetilde O_{p} := & \, \sum_{x\in[n]}|f_x\>\<f_x|_\regT \otimes \sum_{\beta\in\{0,1\}} 
( K_{x,\beta})_\regQ \otimes |\beta\>_{\regR_t},
\end{align*}
which, as before, acts as the identity on the earlier record subregisters $\regR_1\ldots \regR_{t-1}$.

Recall that $|\unif\>=\sum_{x\in[n]}|f_x\>/\sqrt{n}$, $\Pi_{succ}=\sum_{x\in[n]}|f_x,x\>\<f_x,x|$, and that $|\psi^0\>$ is the initial state of the computation. As before, the initial state of the extended computation is $|\phi_0\>=|\unif\>\otimes|\psi^0\>$, the intermediate states just before the extended oracle calls and the final state are defined as $|\phi_t\>:=U_t\widetilde O_{p}|\phi_{t-1}\>$, where $t\in\{1,\ldots,\tau\}$, and the success probability for a randomly chosen $f_x$ is $q_{succ}=\|\Pi_{succ}|\phi_\tau\>\|^2$.

\paragraph{Progress-defining subspaces.}

Now, similarly as in Section~\ref{ssec:progressSubsp}, when defining the progress measure, if any of the registers $\regR_t$ will be in state $|1\>$, we will intuitively assume that the extended computation has found the correct solution.
With this in mind, for the record register $\regR=\regR_1\ldots \regR_{t}$ containing $t$ entries, we define complementary projectors
$\Lambda^{\progA\progB}_t:= \bigotimes_{s=1}^t 
|0\>\<0|_{\regR_s}$
and $\Lambda^{\progC}_t:= I_\regR - \Lambda^{\progA\progB}_t$.
Then, given these new definitions of $\Lambda^{\progA\progB}_t$ and $\Lambda^{\progC}_t$, we define the progress projectors on registers $\regFunc\regR$ exactly as before:
\begin{align*}
&
\Pi^{\progC}_{t} := 
 I_\regT\otimes \Lambda^\progC_t,
\\ &
\Pi^{\progB}_{t} := 
 (I_\regT-|\unif\>\<\unif|) \otimes \Lambda^{\progA\progB}_t,
\\ &
\Pi^{\progA}_{t} :=  
 |\unif\>\<\unif| \otimes \Lambda^{\progA\progB}_t.
\end{align*}
As for the progress measure, we now define it as
\[
\Psi_t:=\|\Pi^\progC_{t}|\phi_t\>\|^2 + 2 \|\Pi^\progB_{t}|\phi_t\>\|^2.
\]

\bigskip

Similarly to Lemma~\ref{lem:progEvol} of the main text, we now can show the following.

\begin{lem}
\label{lem:progEvolNegl}
We have 
\begin{subequations}
\begin{align}
& \Psi_0 = 0, \label{lem:progEvolNegl:A} \\
& q_{succ} \le \Psi_\tau +  \frac2{n}, \label{lem:progEvolNegl:B} \\
& \Psi_{t+1}-\Psi_{t}       \le \frac{28(1-p)}{np} . \label{lem:progEvolNegl:C}
\end{align}
\end{subequations}
\end{lem}

The first two claims of Lemma~\ref{lem:progEvolNegl}, namely, (\ref{lem:progEvolNegl:A}) and (\ref{lem:progEvolNegl:B}), are proven exactly the same way as the first two claims of Lemma~\ref{lem:progEvol}. The proofs of the last claims of the two lemmas also share many similarities, yet there are enough differences to warrant presenting the proof of inequality (\ref{lem:progEvolNegl:C}) in full.

The main difference, informally speaking, is that, unlike the noisy oracle considered in the main text, the negligent oracle with the noise rate $p$ close to $1$ can transfer only very little probability weight to the subspaces $\cH^\progB$ and $\cH^\progC$. In an extreme case when $p=1$, the oracle acts as the identity and those subspaces cannot be reached at all.

We prove  (\ref{lem:progEvolNegl:C}) in Sections~\ref{ssec:negl:change} and \ref{sec:ProofsOfClaims:negl}. Before that, let us note that the proof of Theorem~\ref{thm:negligent} given Lemma~\ref{lem:progEvolNegl} proceeds in essentially the same way as the proof of Theorem~\ref{thm:mainSingle} given Lemma~\ref{lem:progEvol}. In particular, we can now observe that the success probability $q_{succ}$ is at most 
\[
\frac{2}{n}+\frac{28\tau(1-p)}{np} \le \frac{28(\tau(1-p)+1)}{np}.
\]

\subsection{Change in the progress measure}
\label{ssec:negl:change}

Recall the definitions of $|u_x\>$ and $|\tilde f_x\>$ from Section~\ref{ssec:ActAndPas}, which were introduced when analyzing the scenario of noisy target qubit:
\[
|\unif_{x}\>=\frac{1}{\sqrt{n-1}}\sum_{x'\in[n]\setminus\{x\}}|f_{x'}\>,
\qquad
|\tilde{f}_{x}\> =  \sqrt{\frac{n-1}{n}}|f_x\> - \frac1{\sqrt{n}}|\unif_x\>
=  \frac{\sqrt{n}|f_x\>-|\unif\>}{\sqrt{n-1}}.
\]
Recall also that $|\tilde{f}_x\>$ and $|\unif\>$ are orthogonal, and thus
$|\unif_{x}\>\<\unif_{x}|+|f_x\>\<f_x|=|\unif\>\<\unif|+|\tilde f_x\>\<\tilde f_x|$. 
We decompose $\Pi^\progB_t\otimes I_\regQi = \Pi^{\progB,\actv}_{t} + \Pi^{\progB,\pasv}_{t}$, where
\begin{align*}
 \Pi^{\progB,\actv}_{t} := & \sum_{x\in[n]}|\tilde f_x\>\<\tilde f_x| \otimes |x,1\>\<x,1| \otimes \Lambda^{\progA\progB}_{t},  \\
 \Pi^{\progB,\pasv}_{t} := & \sum_{x\in[n]}
\big(I_\regT-|\unif_{x}\>\<\unif_{x}|-|f_x\>\<f_x|\big) \otimes |x,1\>\<x,1| \otimes \Lambda^{\progA\progB}_{t}
\\ &
+\big(I_\regT-|\unif\>\<\unif|\big)\otimes I_\regQi\otimes|0\>\<0| \otimes \Lambda^{\progA\progB}_{t}.
\end{align*}
The following Claims~\ref{clm:CFixed:negl}--\ref{clm:targ:toProgC:negl} for the negligent oracle correspond respectively to Claims~\ref{clm:CFixed}--\ref{clm:targ:toProgC} for the noisy oracle considered in the main text.

\begin{clm}
\label{clm:CFixed:negl}
We have $\widetilde O_{p} \Pi^{\progC}_{t} =  \Pi^{\progC}_{t+1} \widetilde O_{p} \Pi^{\progC}_{t}$, and the images of 
  $\Pi^{\progC}_{t+1} \widetilde O_{p} (\Pi^{\progB}_{t}+\Pi^{\progA}_{t})$
 and
  $\Pi^{\progC}_{t+1} \widetilde O_{p} \Pi^{\progC}_{t}$
   are orthogonal.
\end{clm}

\begin{clm}
\label{clm:OQonPas:negl}
We have $\widetilde O_{p} \Pi^{\progB,\pasv}_{t} =  \Pi^{\progB,\pasv}_{t+1} \widetilde O_{p}$.
\end{clm}

\begin{clm}
 \label{clm:targ:ProgAtoProgB:negl}
We have both
$\| \Pi^{\progB,\actv}_{t+1} \widetilde O_p \Pi^{\progA}_{t} \| = \frac{2(1-p)\sqrt{n-1}}{n}$ and
$  \|\Pi^{\progB,\actv}_{t+1} \widetilde O_p \Pi^{\progB,\actv}_{t} \|
  = |1 - 2(1-p)(1-1/n) |$.
\end{clm}

\begin{clm}
\label{clm:targ:toProgC:negl}
 We have both 
 $\| \Pi^{\progC}_{t+1} \widetilde O_{p} \Pi^{\progB,\actv}_{t}\| = 2\sqrt{p(1-p)(1-1/n)}$
  and 
$\| \Pi^{\progC}_{t+1} \widetilde O_{p} \Pi^{\progA}_{t}\|= 2\sqrt{p(1-p)/n}$.
 \end{clm}

Given these claims, which we will prove in Section~\ref{sec:ProofsOfClaims:negl}, we can proceed to prove the inequality (\ref{lem:progEvolNegl:C}).
Similarly as in Section~\ref{ssec:transitions}, we have
\begin{align*}
\Psi_{t+1}  - \Psi_t  = \, &
\|\Pi^\progC_{t+1} \widetilde O_p(\Pi^{\progB,\actv}_{t}+\Pi^{\progA}_{t})|\phi_t\>\|^2
-2 \|\Pi^{\progB,\actv}_{t} |\phi_t\>\|^2
\\ & 
+2 \|\Pi^{\progB,\actv}_{t+1} \widetilde O_p (\Pi^{\progB,\actv}_{t}+\Pi^{\progA}_{t}) |\phi_t\>\|^2
\\ 
\le \, & 
\Big( \|\Pi^\progC_{t+1} \widetilde O_p \Pi^{\progB,\actv}_{t}\|
  \!\cdot\!  \|\Pi^{\progB,\actv}_{t}|\phi_t\>\|
  +  \|\Pi^\progC_{t+1} \widetilde O_p \Pi^{\progA}_{t}\|\Big)^2
-2 \|\Pi^{\progB,\actv}_{t} |\phi_t\>\|^2
\\ & 
+2 \Big(
 \|\Pi^{\progB,\actv}_{t+1} \widetilde O_p \Pi^{\progB,\actv}_{t}\|\!\cdot\! \|\Pi^{\progB,\actv}_{t} |\phi_t\>\| +   \|\Pi^{\progB,\actv}_{t+1} \widetilde O_p \Pi^{\progA}_{t}\|
\Big)^2
%
   \\ = \, &
 \frac{4(1-p)}{n} \bigg[
-p
\underbrace{\Big(1+2\frac{1-p}{np}\Big)}_{\ge 1} \Big(\sqrt{n-1}\|\Pi^{\progB,\actv}_{t} |\phi_t\>\|\Big)^2
\\ &
+ 2\Big(\underbrace{p + \big|1 - 2(1-p)(1-1/n)\big| }_{\le 2}\Big)
\Big(\sqrt{n-1}\|\Pi^{\progB,\actv}_{t} |\phi_t\>\|\Big)
\\ & +
  \underbrace{p + 2 (1-p)(1-1/n)}_{\le 3 \le 3/p }\bigg].
  \end{align*}
  Here the first equality is due to Claims~\ref{clm:CFixed:negl} and \ref{clm:OQonPas:negl}, then the triangle inequality has been applied, and the final equality is due to Claims~\ref{clm:targ:ProgAtoProgB:negl} and \ref{clm:targ:toProgC:negl} followed by somewhat tedious yet straight forward simplifications.
 Note that,  now, unlike in the main text, after applying the triangle inequality to $\|\Pi^\progC_{t+1} \widetilde O_p(\Pi^{\progB,\actv}_{t}+\Pi^{\progA}_{t})|\phi_t\>\|$, we do not subsequently apply the inequality $(N_1+N_2)^2\le 2(N_1^2+N_2^2)$. The reason for that is that here the factor $2$ that would be introduced by the latter inequality is too large to obtain the desired result.

From the above bound on $\Psi_{t+1}  - \Psi_t $, 
by denoting $\alpha:=\sqrt{n-1}\|\Pi^{\progB,\actv}_{t} |\phi_t\>\|$ for the sake of brevity and by applying the inequalities under the braces, we further get
\begin{align*}
\Psi_{t+1}  - \Psi_t  \le \, &
 \frac{4(1-p)}{n} \bigg[
-p \alpha^2
+ 4 \alpha
+ \frac3p \bigg]
      = 
\frac{4(1-p)}{np} \bigg[ 7-(p\alpha + 2)^2  \bigg]
      \le
\frac{28(1-p)}{np} ,
  \end{align*}
which is the claimed inequality (\ref{lem:progEvolNegl:C}). This concludes the proof of Lemma~\ref{lem:progEvolNegl}, given Claims~\ref{clm:CFixed:negl}--\ref{clm:targ:toProgC:negl}. We prove those claims next.

%
%
%
%

\subsection{Proofs of Claims~\ref{clm:CFixed:negl}--\ref{clm:targ:toProgC:negl}}
\label{sec:ProofsOfClaims:negl}

\noindent
\textbf{Claim~\ref{clm:CFixed:negl}} (restated)\textbf{.}
\emph{
We have $\widetilde O_{p} \Pi^{\progC}_{t} =  \Pi^{\progC}_{t+1} \widetilde O_{p} \Pi^{\progC}_{t}$, and the images of 
  $\Pi^{\progC}_{t+1} \widetilde O_{p} (\Pi^{\progB}_{t}+\Pi^{\progA}_{t})$
 and
  $\Pi^{\progC}_{t+1} \widetilde O_{p} \Pi^{\progC}_{t}$
   are orthogonal.
} \bigskip

The proof of Claim~\ref{clm:CFixed:negl} is essentially equivalent to that of Claim~\ref{clm:CFixed}, except that now the record register does not store any labels of Pauli operators.

\bigskip

For the remaining claims, note that, using the above expressions for Kraus operators $K_{x,0}$ and $K_{x,1}$, we can write $\widetilde O_{p}$ as
\[
\widetilde O_{p} = 
I_{\regT\regQ}\otimes |0\>_{\regR_t}
+ 2 \sum_{x\in[n]}
|f_x,x,1\>\<f_x,x,1|_{\regT\regQ}\otimes 
\big(-(1-p)|0\> + \sqrt{p(1-p)}|1\>\big)_{\regR_t}.
\] 
The proof of Claim~\ref{clm:OQonPas:negl} proceeds along similar lines as that of Claim~\ref{clm:OQonPas}.

 \bigskip
\noindent
\textbf{Claim~\ref{clm:OQonPas:negl}} (restated)\textbf{.}
\emph{
We have $\widetilde O_{p} \Pi^{\progB,\pasv}_{t} =  \Pi^{\progB,\pasv}_{t+1} \widetilde O_{p}$.
}

\begin{proof} 
Note that $\Pi_{t+1}^{\progB,\pasv}=\Pi_{t}^{\progB,\pasv}\otimes|0\>\<0|$, and consider the above expression for $\widetilde O_p$.
We have both that $|f_x,x,1\>\<f_x,x,1|\Pi_{t}^{\progB,\pasv}=0$ and that $I_{\regT\regQ}$ commutes with $\Pi_{t}^{\progB,\pasv}$, which prove the claim.
\end{proof}

In the following two proofs, as in the main text, let $\Pi^{\mathfrak{Lab}}_t\in\big\{\Pi^{\progA}_t,\Pi^{\progB,\actv}_t\big\}$.

\bigskip
\noindent
\textbf{Claim~\ref{clm:targ:ProgAtoProgB:negl}} (restated)\textbf{.}
\emph{
We have both
$\| \Pi^{\progB,\actv}_{t+1} \widetilde O_p \Pi^{\progA}_{t} \| = \frac{2(1-p)\sqrt{n-1}}{n}$ and
$  \|\Pi^{\progB,\actv}_{t+1} \widetilde O_p \Pi^{\progB,\actv}_{t} \|
  = |1 - 2(1-p)(1-1/n) |$.
}

\begin{proof}
By taking the part of $\widetilde O_{p}$ that appends $|0\>_{\regR_{t+1}}$ to the record register,  we have
\begin{align*}
\| \Pi^{\progB,\actv}_{t+1}  \widetilde O_{p} \Pi_t^{\progLab}\|
 & =
\bigg\|
\Pi^{\progB,\actv}_{t} 
\Big(I_{\regT\regQ} - 2(1-p) \sum_{x'\in[n]} |f_{x'},x',1\>\<f_{x'},x',1|\Big)
  \Pi_t^{\progLab}\bigg\|.
  \end{align*}
For any given $x\in[n]$ and $\beta\in\{0,1\}$, we note that $|x,\beta\>\<x,\beta|_\regQ$ commutes with all $\Pi^{\progB,\actv}_{t}$, $\Pi_t^{\progLab}$, and the operator in the large parentheses above.
Hence, due to symmetry, the norm above equals
\begin{align*}
\Big\|\big(\<x,1|\otimes I\big)
\big(\Pi^{\progB,\actv}_{t}\Pi_t^{\progLab} 
 - 2(1-p) 
 \Pi^{\progB,\actv}_{t}|f_x,x,1\>\<f_x,x,1|
  \Pi_t^{\progLab}
 \big)\big(|x,1\>\otimes I\big) \Big\|.
\end{align*}
For $\Pi^{\mathfrak{Lab}}_t=\Pi^{\progA}_t$, the term $\Pi^{\progB,\actv}_{t}\Pi_t^{\progA}$  vanishes, and we get
\begin{multline*}
\| \Pi^{\progB,\actv}_{t+1}  \widetilde O_{p} \Pi_t^{\progA}\|
  =
2(1-p)
\big\| \Pi^{\progB,\actv}_{t}|f_x,x,1\> \big\|
\!\cdot\!
\big\| \<f_x,x,1| \Pi_t^{\progA}\big\|
\\  =
2(1-p)\big|\<\tilde{f}_{x}|f_x\>\big|\!\cdot\!\big|\<f_x| \unif\> \big| = 2(1-p)\frac{\sqrt{n-1}}{n}.
  \end{multline*}
For $\Pi^{\mathfrak{Lab}}_t=\Pi^{\progB,\actv}_t$, we get
\begin{align*}
\| \Pi^{\progB,\actv}_{t+1}  \widetilde O_{p} \Pi_t^{\progB,\actv}\|
 & =
\big\|
|\tilde{f}_{x}\>\<\tilde{f}_{x}|
 - 2(1-p)  |\tilde{f}_{x}\>\<\tilde{f}_{x}|f_x\>\<f_x|\tilde{f}_{x}\>\<\tilde{f}_{x}|
\big\|
\\  & =
\Big| 1 - 2(1-p)\frac{n-1}{n} \Big|. \qedhere
  \end{align*}
\end{proof}

\noindent
\textbf{Claim~\ref{clm:targ:toProgC:negl}} (restated)\textbf{.}
\emph{
 We have both 
 $\| \Pi^{\progC}_{t+1} \widetilde O_{p} \Pi^{\progB,\actv}_{t}\| = 2\sqrt{p(1-p)(1-1/n)}$
  and 
$\| \Pi^{\progC}_{t+1} \widetilde O_{p} \Pi^{\progA}_{t}\|= 2\sqrt{p(1-p)/n}$.
}

\begin{proof}
By taking the part of $\widetilde O_{p}$ that appends $|1\>_{\regR_{t+1}}$ to the record register,  we have
\begin{align*}
\| \Pi^{\progC}_{t+1}  \widetilde O_{p} \Pi_t^{\progLab}\|
 & =
 2 \sqrt{p(1-p)}
\bigg\|
\Pi^{\progC}_{t+1} 
\Big( \sum_{x\in[n]} |f_{x},x,1\>\<f_{x},x,1|\otimes|1\>_{\regR_{t+1}}\Big)
  \Pi_t^{\progLab}\bigg\|
  \\ & =
 2 \sqrt{p(1-p)}
\bigg\|
\sum_{x\in[n]} |f_{x},x,1\>\<f_{x},x,1|
  \Pi_t^{\progLab}\bigg\|.
  \end{align*}
Since $|x,1\>\<x,1|_\regQ$ commutes with $\Pi_t^{\progLab}$, due to symmetry among $x\in[n]$, for any such $x$ we have
\begin{align*}
\| \Pi^{\progC}_{t+1}  \widetilde O_{p} \Pi_t^{\progLab}\|
 & =
 2 \sqrt{p(1-p)}
\big\|\<f_{x},x,1|\Pi_t^{\progLab} \big\|.
  \end{align*}
We conclude by observing that
\begin{align*}
& \big\|\<f_{x},x,1|\Pi_t^{\progB,\actv} \big\| = \big|\<f_x|\tilde f_x\> \big| = \sqrt{1-1/n},
\\ & \big\|\<f_{x},x,1|\Pi_t^{\progA} \big\| = \big|\<f_x|u\> \big| = 1/\sqrt{n}. \qedhere
\end{align*}
\end{proof}

\section{One-Sided Noise Versus Two-Sided Noise}
\label{sec:OneTwoSided}

Let us show that the $\Omega(nr)$ lower bound given by Theorem~\ref{thm:mainSingle} does not hold when the noise is one-sided or when we have flag bits indicating the presence of error.
In all cases, we assume that $j$ is known in advance; otherwise, for the upper bounds, we get a factor-$\log n$ slowdown by having to guess $j$.

We are rather informal with the algorithms that we provide. All of them are essentially versions of Grover's algorithm adapted to specific scenarios.

Unlike when proving the lower bound (i.e., Theorem~\ref{thm:mainSingle}), for upper bounds, we assume that there might be multiple marked elements. This assumption slightly complicates some arguments,
in particular, when the index-qubit noise is concerned. For that matter, we will find the following simple claim useful.

\begin{clm}
\label{clm:coins}
Suppose we are tossing an unbiased coin until we have an odd number of heads and an even number of tails. The expected number of coin tosses $T$ until this happens is $3$.
\end{clm}

\begin{proof}
Let us inspect the first coin toss. With probability $1/2$ it is heads, and we are done. Otherwise, it is tails, and we have two scenarios depending on the second toss.

If the second toss comes out as tails, we are back to the starting position, where we require odd number of heads and an even number of tails. If the second toss comes out as heads, we are at an equivalent position where we require odd number of tails and an even number of heads.

Hence, $T$ satisfies $T = \frac12\cdot 1 + \frac12\cdot(2+T)$. This solves as $T=3$.
\end{proof}

\subsection{Overcoming noise after the oracle call}

Here let us consider the noisy oracle call to be $\ccO_f\circ\ccN_{j,r}$ for known $j$ (see Figure~\ref{fig:noisy_qubit_orac_after} for illustrations). As explained below, if $j=0$, we do not have any slowdown compared to Grover's algorithm, while, if $j\in\{1,\ldots,\log n\}$, we experience slowdown of a factor about $\sqrt{2}$, as we are searching over two spaces of half the size.

\begin{figure}[!h]
\centering
\begin{tikzpicture}
\draw [white] (-0.1,-0.35-\ygap*2.5) rectangle (0.1,0.25+\ygap*4); 
  \foreach \y in {1,...,5}
  {
      \node [gray] at (0,\y*\ygap-\ygap) {\tiny $\y$};
   }
   \node [gray] at (0,-2.5*\ygap) {\tiny $0$};
\end{tikzpicture}
\hspace{13pt}
\begin{tikzpicture}
\draw [white] (-2.4,-0.35-\ygap*2.5) rectangle (2.6,0.25+\ygap*4); 
  \foreach \y in {1,...,5}
  {
      \draw (-1.2,\ygap*\y-\ygap)--(2.4,\ygap*\y-\ygap) ;
   }
   \draw (-1.2,-\ygap*2.5)--(2.4,-\ygap*2.5);
   \draw [draw=black,fill=oracleColor] (-0.5,-0.15-\ygap*2.5) rectangle (0.5,0.15+\ygap*4);
   \draw [fill=noiseColor, rounded corners = 1.3mm] (0.8,-0.26-2.5*\ygap) rectangle (1.8,0.26-2.5*\ygap);

   \node at (0,\ygap) {$O_f$};
    \node at (1.32,-.03-2.5*\ygap) {$\ccN_{0,r}$};
\end{tikzpicture}
\hspace{17pt}
\begin{tikzpicture}
\draw [white] (-2.4,-0.35-\ygap*2.5) rectangle (2.6,0.25+\ygap*4); 
  \foreach \y in {1,...,5}
  {
      \draw (-1.2,\ygap*\y-\ygap)--(2.4,\ygap*\y-\ygap) ;
   }
   \draw (-1.2,-\ygap*2.5)--(2.4,-\ygap*2.5);
   \draw [draw=black,fill=oracleColor] (-0.5,-0.15-\ygap*2.5) rectangle (0.5,0.15+\ygap*4);
   \draw [fill=noiseColor, rounded corners = 1.3mm] (0.8,-0.26+3*\ygap) rectangle (1.8,0.26+3*\ygap);

   \node at (0,\ygap) {$O_f$};
    \node at (1.32,-.03+3*\ygap) {$\ccN_{4,r}$};
\end{tikzpicture}

{\hspace{29pt}\small(\emph{target-qubit noise}) \hspace{78pt}
\small(\emph{index-qubit noise})}
\captionsetup{font=small}
\captionsetup{width=0.9\textwidth}
\caption[my caption]{%
The circuit diagrams of the noisy oracle call when the depolarizing noise $\ccN_{j,r}$ acts only after the faultless oracle call $O_f$.}
\label{fig:noisy_qubit_orac_after}
\end{figure}
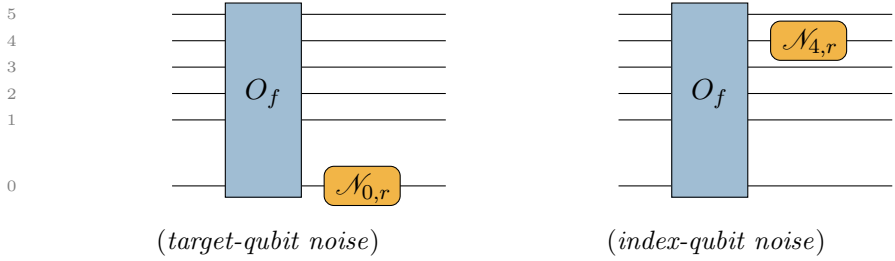

\paragraph{Noisy target qubit, $j=0$.}

To deal with the noise, we run Grover's algorithm, except that, before every oracle call, we initialize the target qubit to a new instance of $|1\>$, perform the \emph{noisy} oracle call, and then discard the target qubit, which may or may not have been affected by the noise after the application of $O_f$.

\paragraph{Noisy index qubit, $j\in\{1,\ldots,\log n\}$.}

In this scenario, we run two instances of Grover's algorithm, one searching over all $x$ such that $x_{\mathrm b,j}=0$ and one over all $x$ such that $x_{\mathrm b,j}=1$. Before every oracle call, we reinitialize $x_{\mathrm b,j}$ to the desired value, so no problem arrises even if the depolarizing noise $\ccN_{j,r}$ randomizes this value after $O_f$.

\subsection{Overcoming noise before the oracle call}
\label{ssec:noiseBefore}

Now let us consider the noisy oracle call to be $\ccN_{j,r}\circ \ccO_f$, that is, it is composed of the depolarizing noise on $j$-th query qubit followed by the faultless oracle call (see Figure~\ref{fig:noisy_qubit_orac_before} for illustrations). For sake of simplicity, in both scenarios, let us purposely completely depolarize $j$-th query qubit before every oracle call. Because $\ccN_{j,1}\circ (\ccN_{j,r}\circ \ccO_f)=\ccN_{j,1}\circ \ccO_f$, this is equivalent to having $r=1$.

\begin{figure}[!h]
\centering
\begin{tikzpicture}
\draw [white] (-0.1,-0.35-\ygap*2.5) rectangle (0.1,0.25+\ygap*4); 
  \foreach \y in {1,...,5}
  {
      \node [gray] at (0,\y*\ygap-\ygap) {\tiny $\y$};
   }
   \node [gray] at (0,-2.5*\ygap) {\tiny $0$};
\end{tikzpicture}
\hspace{13pt}
\begin{tikzpicture}
\draw [white] (-2.4,-0.35-\ygap*2.5) rectangle (2.6,0.25+\ygap*4); 
  \foreach \y in {1,...,5}
  {
      \draw (-2.4,\ygap*\y-\ygap)--(1.2,\ygap*\y-\ygap) ;
   }
   \draw (-2.4,-\ygap*2.5)--(1.2,-\ygap*2.5);
   \draw [draw=black,fill=oracleColor] (-0.5,-0.15-\ygap*2.5) rectangle (0.5,0.15+\ygap*4);
   \draw [fill=noiseColor, rounded corners = 1.3mm] (-1.8,-0.26-2.5*\ygap) rectangle (-0.8,0.26-2.5*\ygap);

   \node at (0,\ygap) {$O_f$};
    \node at (-1.28,-.03-2.5*\ygap) {$\ccN_{0,r}$};
\end{tikzpicture}
\hspace{17pt}
\begin{tikzpicture}
\draw [white] (-2.4,-0.35-\ygap*2.5) rectangle (2.6,0.25+\ygap*4); 
  \foreach \y in {1,...,5}
  {
      \draw (-2.4,\ygap*\y-\ygap)--(1.2,\ygap*\y-\ygap) ;
   }
   \draw (-2.4,-\ygap*2.5)--(1.2,-\ygap*2.5);
   \draw [draw=black,fill=oracleColor] (-0.5,-0.15-\ygap*2.5) rectangle (0.5,0.15+\ygap*4);
   \draw [fill=noiseColor, rounded corners = 1.3mm] (-1.8,-0.26+3*\ygap) rectangle (-0.8,0.26+3*\ygap);

   \node at (0,\ygap) {$O_f$};
    \node at (-1.28,-.03+3*\ygap) {$\ccN_{4,r}$};
\end{tikzpicture}

{\hspace{29pt}\small(\emph{target-qubit noise}) \hspace{78pt}
\small(\emph{index-qubit noise})}
\captionsetup{font=small}
\captionsetup{width=0.9\textwidth}
\caption[my caption]{%
The circuit diagrams of the noisy oracle call when the depolarizing noise $\ccN_{j,r}$ acts only before the faultless oracle call $O_f$.}
\label{fig:noisy_qubit_orac_before}
\end{figure}
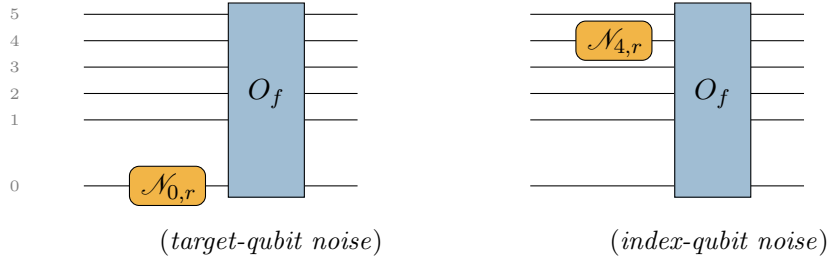

\paragraph{Noisy target qubit, $j=0$.}

Let us describe a procedure that, using the noisy oracle call $\ccN_{0,1}\circ \ccO_f$ the expected number of $T=2$ times acts on the query-input register $\regQi$ as the reflection about the marked elements.
The procedure calls the noisy oracle and then measures the target qubit. If the measurement yields $|1\>$, which happens with probability $1/2$, the desired reflection has been performed correctly. If the measurement yields $|0\>$, the noisy oracle call had no effect on the query-input register $\regQi$, and we can again call the noisy oracle and try to measure the target register to $|1\>$. (Note that $T=\frac12\cdot 1 + \frac12\cdot(1+T)$ indeed solves to $T=2$.)

After we have successfully performed the reflection about the marked elements, we can proceed with the algorithm by applying Grover's diffusion operator. 
Since the expected number of noisy oracle calls to implement a faultless reflection is $2$, the expected number of noisy oracle calls to implement $\Theta(\sqrt{n})$ faultless reflections is also $\Theta(\sqrt{n})$. Hence, by Markov's inequality, the probability that it requires $\omega(\sqrt{n})$ noisy oracle calls is negligible, i.e., it is $o(1)$.

\paragraph{Noisy index qubit, $j\in\{1,\ldots,\log n\}$.}

Let us show how to search for a marked element $x$ such that $x_{\bin,j}=0$, the case $x_{\bin,j}=1$ being analogous.
We initialize the target qubit $\regQo$ to $|1\>$, and its state will remain $|1\>$ throughout the computation.
Also throughout the computation, $j$-th query qubit will remain unentangled to any other qubit.

Let $y_\bin$ denote an $(\log n -1)$-bit string, which we think of as a $\log n$-bit string with the $j$-th bit missing. Then, for $\beta\in\{0,1\}$, let $y_\bin\|_j\beta$ denote the $\log n$-bit string that is obtained form $y_\bin$ by inserting $\beta$ as the $j$-th bit.
We want to implement the reflection about all $y_\bin$ such that (the input corresponding to the bit string) $y_\bin\|_j 0$ is marked, which in essence is the reflection  about all marked $x$ such that $x_{\bin,j}=0$.
Let us denote this reflection, which acts on all the qubits of register $\regQi$ except the $j$-th qubit, by $R_{j,0}$, and, similarly, let $R_{j,1}$ denote the reflection about all $y_\bin$ such that $y_\bin\|_j 1$ is marked.

To perform the desired reflection $R_{j,0}$, we repeatedly call the noisy oracle $\ccN_{j,1}\circ O_f$ and measure $j$-th query qubit until that measurement has yielded $|0\>$ an odd number of times and $|1\>$ an even number of times. According to Claim~\ref{clm:coins}, the expected number of noisy oracle calls to achieve this is $T=3$.
Note that, whenever we measure $|1\>$, we have just performed the reflection $R_{j,1}$. However, because we have an even number of such reflections and because they commute with $R_{j,0}$, they pairwise cancel out. As for the reflections $R_{j,0}$, we have an odd number of them, so all but one cancels out.

The rest of the analysis for the search of a marked $x$ such that $x_{\bin,j}=0$ goes along the same lines as in the scenario of the noisy target qubit. That is, using the linearity of expectation and Markov's inequality, we can see that $\OO(\sqrt n)$ calls to the noisy oracle suffice to perform this search.

\subsection{Search with flag bits}
\label{sec:search_with_flags}

Let us again assume that the noise is two-sided, occurring both before and after the faultless oracle call $O_f$. Except now, unlike in the main text, we assume that the noise operation, which we will denote by $\ccN^+_{j,r}$, produces a flag bit indicating whether or not the completely depolarizing error has occurred. Here we show that having such flag bits reduce the complexity from $\Omega(\max\{nr,\sqrt{n}\})$  to $\OO(\max\{nr^2,\sqrt{n}\})$. This contrasts the noise which simultaneously depolarizes all the query register $\regQ$ studied in~\cite{Rosmanis:2023:NoisyOracle}; there having flag bits had no effect on the complexity.

More formally, we define the \emph{error-signaling} noise as
\[
\ccN^+_{j,r} := 
(1-r)\itangle{|} \mathit{0} \itangle{\rangle}\otimes\ccI
+ r\itangle{|} \mathit{1} \itangle{\rangle}\otimes\ccN_{j,1},
\]
where $\itangle{|} \mathit{0} \itangle{\rangle}$ and $\itangle{|} \mathit{1} \itangle{\rangle}$ denote, respectively, the CPTP maps that prepare the states $|0\>$ and $|1\>$ to be used as a flag.
Note that, in contrast, the \emph{error-concealing} noise considered in the main text is
\[
\ccN_{j,r} = (1-r)\ccI + r \ccN_{j,1}.
\]
We define the error-signaling noisy oracle call as $\ccO^+_{f,j,r}:=\ccN^+_{j,r}\circ\ccO_f\circ\ccN^+_{j,r}$ (see Figure~\ref{fig:noisy_qubit_orac_signal} for illustrations), and note that this oracle call introduces a pair of flag bits $(\beta,\beta')$, the former indicating whether the error occurred before $O_f$, the latter indicating whether the error occurred after $O_f$.

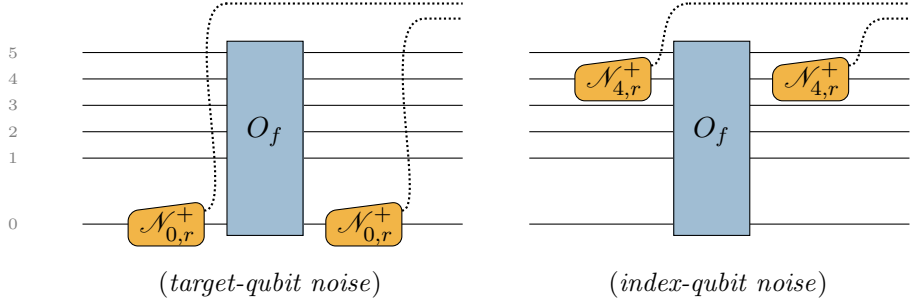
\begin{figure}[!h]
\centering
\begin{tikzpicture}
\draw [white] (-0.1,-0.4-\ygap*2.5) rectangle (0.1,1.0+\ygap*3); 
  \foreach \y in {1,...,5}
  {
      \node [gray] at (0,\y*\ygap-\ygap) {\tiny $\y$};
   }
   \node [gray] at (0,-2.5*\ygap) {\tiny $0$};
\end{tikzpicture}
\hspace{13pt}
\begin{tikzpicture}
\draw [white] (-2.4,-0.4-\ygap*2.5) rectangle (2.6,1.0+\ygap*3); 
  \foreach \y in {1,...,5}
  {
      \draw (-2.4,\ygap*\y-\ygap)--(2.6,\ygap*\y-\ygap) ;
   }
   \draw (-2.4,-\ygap*2.5)--(2.6,-\ygap*2.5);
   \draw [draw=black,fill=oracleColor] (-0.5,-0.15-\ygap*2.5) rectangle (0.5,0.15+\ygap*4);
   \draw [fill=noiseColor, rounded corners = 1.3mm] 
   (-0.8,0.30-2.5*\ygap) -- (-0.8,-0.29-2.5*\ygap) -- (-1.8,-0.29-2.5*\ygap) -- (-1.8, 0.17-2.5*\ygap) -- cycle;
   \draw [fill=noiseColor, rounded corners = 1.3mm] 
   (1.8,0.30-2.5*\ygap) -- (1.8,-0.29-2.5*\ygap) -- (0.8,-0.29-2.5*\ygap) -- (0.8, 0.17-2.5*\ygap) -- cycle;

   \draw [densely dotted, thick]  (-0.8,0.18-2.5*\ygap) .. controls 
             (-0.35,0.18-2.5*\ygap) and
             (-1.15,1.00+3.0*\ygap) ..
             (-0.5,1.00+3.0*\ygap);
    \draw [densely dotted, thick] (-0.52,1.00+3*\ygap)--(2.6,1.00+3*\ygap) ;

\draw [densely dotted, thick] (1.8,0.18-2.5*\ygap) .. controls 
             (2.15,0.18-2.5*\ygap) and
             (1.45,0.80+3*\ygap) ..
             (2.1,0.80+3*\ygap);
 \draw [densely dotted, thick]  (2.1,0.80+3*\ygap)--(2.6,0.80+3*\ygap) ;

   \node at (0,\ygap) {$O_f$};
    \node at (-1.28,-0.025-2.5*\ygap) {$\ccN_{0,r}^+$};
    \node at (1.32,-0.025-2.5*\ygap) {$\ccN_{0,r}^+$};

\end{tikzpicture}
\hspace{17pt}
\begin{tikzpicture}
\draw [white] (-2.4,-0.4-\ygap*2.5) rectangle (2.6,1.0+\ygap*3); 
  \foreach \y in {1,...,5}
  {
      \draw (-2.4,\ygap*\y-\ygap)--(2.6,\ygap*\y-\ygap) ;
   }
   \draw (-2.4,-\ygap*2.5)--(2.6,-\ygap*2.5);
   \draw [draw=black,fill=oracleColor] (-0.5,-0.15-\ygap*2.5) rectangle (0.5,0.15+\ygap*4);
   \draw [fill=noiseColor, rounded corners = 1.3mm] 
   (-0.8,0.30+3*\ygap) -- (-0.8,-0.29+3*\ygap) -- (-1.8,-0.29+3*\ygap) -- (-1.8, 0.17+3*\ygap) -- cycle;
   \draw [fill=noiseColor, rounded corners = 1.3mm] 
   (1.8,0.30+3*\ygap) -- (1.8,-0.29+3*\ygap) -- (0.8,-0.29+3*\ygap) -- (0.8, 0.17+3*\ygap) -- cycle;

   \draw [densely dotted, thick]  (-0.8,0.18+3*\ygap) .. controls 
             (-0.55,0.18+3*\ygap) and
             (-0.75,1.00+3*\ygap) ..
             (-0.3,1.00+3*\ygap);
    \draw [densely dotted, thick] (-0.3,1.00+3*\ygap)--(2.6,1.00+3*\ygap) ;

\draw [densely dotted, thick] (1.8,0.18+3*\ygap) .. controls 
             (2.05,0.18+3*\ygap) and
             (1.85,0.80+3*\ygap) ..
             (2.3,0.80+3*\ygap);
 \draw [densely dotted, thick]  (2.3,0.80+3*\ygap)--(2.6,0.80+3*\ygap) ;

   \node at (0,\ygap) {$O_f$};
    \node at (-1.28,-0.025+3*\ygap) {$\ccN_{4,r}^+$};
    \node at (1.32,-0.025+3*\ygap) {$\ccN_{4,r}^+$};

\end{tikzpicture}

{\hspace{29pt}\small(\emph{target-qubit noise}) \hspace{78pt}
\small(\emph{index-qubit noise})}
\captionsetup{font=small}
\captionsetup{width=0.9\textwidth}
\caption[my caption]{%
The circuit diagram of the two-sided error-signaling noisy oracle call $\ccO^+_{f,j,r}=\ccN^+_{j,r}\circ\ccO_f\circ\ccN^+_{j,r}$. The flag bits introduced by $\ccN^+_{j,r}$ are displayed as dotted lines.}
\label{fig:noisy_qubit_orac_signal}
\end{figure}

\bigskip

It is natural to ask: Why in the error-signaling case can we get speedups over the error-concealing case, when we do not have flag bits? An answer seems to be: Now, if only one error has occurred, the flag bits let us distinguish whether it occurred before or after $O_f$. As discussed in Section~\ref{sec:intuition}, intuitively, not having such an ability ruins our computation even if only one of the two errors has occurred. On the other hand, now, with the flag bits, the computation gets ruined only if both errors occur.

\bigskip

Below we provide the main ideas behind the algorithms that solve the search using $\OO(\max\{nr^2,\sqrt{n}\})$ calls to the error-signaling noisy oracle $\ccO^+_{f,j,r}$. We consider the scenarios of the noisy target qubit and the noisy index qubit separately.
For both scenarios, we will construct procedures to implement a single iteration of Grover's algorithm that work equally fast when no error occurs or when exactly one of the flag bits is zero. This way we will incur some slowdown compared to more intricate methods, but this slowdown is at most by a constant factor, and the analysis becomes easier.

The framework of the algorithms is essentially the same as that of the algorithm outlined in the original article (see \cite[Section 7]{Rosmanis:2023:NoisyOracle}): in parallel, we run $\Theta(nr^4)$ instances of Grover's algorithm for $\Theta(1/r^2)$ iterations and then, for each instance, we check the correctness of the returned result.

Unlike in the original article, here a single iteration of Grover's algorithm will requite multiple calls to the noisy oracle $\ccN^+_{j,r}$, however, in expectation, still at most a constant number. The advantage that we will have now is that we will not fail whenever just a single error occurs, but only when we get the \emph{double error}, that is, when the error occurs both before and after $O_f$ and, thus, the pair of flag bits given by oracle is $(1,1)$. If we perform $\Theta(1/r^2)$ calls to $\ccN^+_{j,r}$, with at least a constant probability we never get a double error.

If we had no flag bits, checking the correctness of the result of each instance of Grover's algorithm might require a logarithmic number of oracle calls to obtain the desired confidence. However, now, because of the flag bits, in expectation, we can do it using a constant number of oracle calls.

\paragraph{Noisy target qubit, $j=0$.} 

The following procedure, conditioned on the noisy oracle never giving flag bits $(1,1)$, applies on register $\regQi$ the reflection about the marked elements, that is, $I_\regQi-2\sum_{x\colon f(x)=1}|x\>\<x|$.
 
We choose $\beta\in\{0,1\}$ uniformly at random, set the target register to $|\beta\>$, call the noisy oracle $\ccN^+_{0,r}$, and measure the target register in the standard basis, obtaining  $\beta'\in\{0,1\}$.
If we do not get a double error, then we have one of the following two equally likely events. 
\begin{itemize}
\item
If $\beta=1$ and the flag bits are either $(0,0)$ or $(0,1)$ or if $\beta'=1$ and the flag bits are $(1,0)$, then we have applied the desired reflection.
\item
If $\beta=0$ and the flag bits are either $(0,0)$ or $(0,1)$ or if $\beta'=0$ and the flag bits are $(1,0)$, then we have not affected the register $\regQi$ at all.
\end{itemize}
(Note: if the flag bits are $(0,0)$, then $\beta=\beta'$.)
We repeat this subroutine until the former event has occurred, which, in expectation, requires $T=2$ oracle calls.

\paragraph{Noisy index qubit, $j\in\{1,\ldots,\log n\}$.}

Recall the reflections $R_{j,0}$ and $R_{j,1}$ from Section~\ref{ssec:noiseBefore}, which reflect about all $(\log n-1)$-bit strings $y_\bin$ such that, respectively, $y_\bin\|_j 0$ is marked and $y_\bin\|_j 1$ is marked.
The following procedure, conditioned on the noisy oracle never giving flag bits $(1,1)$, applies the reflection $R_{j,0}$, which enables us to search for marked $x$ such that $x_{\bin,j}=0$. The search for marked $x$ such that $x_{\bin,j}=1$ is analogous.

We choose $\beta\in\{0,1\}$ uniformly at random, set the $j$-th query qubit to $|\beta\>$ and the target register to $|1\>$, call the noisy oracle $\ccN^+_{j,r}$, and measure the $j$-th query qubit in the standard basis, obtaining  $\beta'\in\{0,1\}$.
If we do not get a double error, then we have one of the following two equally likely events. 
\begin{itemize}
\item
If $\beta=1$ and the flag bits are either $(0,0)$ or $(0,1)$ or if $\beta'=1$ and the flag bits are $(1,0)$, then we have applied $R_{j,1}$.
\item
If $\beta=0$ and the flag bits are either $(0,0)$ or $(0,1)$ or if $\beta'=0$ and the flag bits are $(1,0)$, then we have applied $R_{j,0}$.
\end{itemize}
We repeat this subroutine until the former event has occurred even number of times and the latter event has occurred odd number of times, which, in expectation, as per Claim~\ref{clm:coins}, requires $T=3$ oracle calls.

\end{document}